\documentclass[11pt,letterpaper]{article}

\usepackage[margin=1in]{geometry}

\usepackage{fullpage}

\usepackage[subtle,mathspacing=normal]{savetrees}
\usepackage{url}
\usepackage{amsmath}
\usepackage{algorithm}
\usepackage[noend]{algpseudocode}
\usepackage{amssymb}
\usepackage{amsthm}
\usepackage{color}
\usepackage{xcolor}
\usepackage{array}
\usepackage{xy}
\usepackage{setspace}
\usepackage{varwidth}% http://ctan.org/pkg/varwidth
\usepackage{multicol}
\usepackage{algorithm}
\usepackage{hyperref}
\usepackage[capitalize]{cleveref}
\usepackage{mathtools}
\usepackage{cite}
\usepackage{thmtools,thm-restate}

\usepackage{comment}
\newtheorem{theorem}{Theorem}[section]

\newtheorem{lemma}[theorem]{Lemma}

\newtheorem{claim}[theorem]{Claim}
\newtheorem{corollary}[theorem]{Corollary}

\theoremstyle{remark}

\crefname{thm}{Theorem}{Theorems}
\crefname{lem}{Lemma}{Lemmas}
\crefname{clm}{Claim}{Claims}
\crefname{obs}{Observation}{Observations}
\crefname{cor}{Corollary}{Corollaries}

\usepackage{authblk}

\usepackage{enumitem}
\setlist[enumerate]{label=\arabic*.}
\newcommand{\defn}[1]{\textbf{\emph{#1}}}
\renewcommand{\paragraph}[1]{\vspace{.2 cm} \noindent \textbf{#1}}

\renewcommand{\epsilon}{\varepsilon}

\newcommand{\E}{\mathbb{E}}

\newcommand{\alg}{\mathcal{A}}
\newcommand{\balg}{\mathcal{B}}
\newcommand{\workload}{W}
\newcommand{\adv}{V}

\newcommand{\hwm}{S}
\newcommand{\Q}{Q}

\newcommand{\size}[1]{|#1|}
\newcommand{\fragnum}{f}
\newcommand{\advmem}{M}
\newcommand{\advreq}{Q}
\newcommand{\M}{M}
\newcommand{\m}{m}
\newcommand{\Mmin}{\M_{\textrm{min}}}
\newcommand{\Mmax}{\M_{\textrm{max}}}
\newcommand{\Mwork}{\M(\workload)}

\newcommand{\Mguess}{\widetilde{\M}}
\newcommand{\Qguess}{{\widetilde{\Q}}}

\newcommand{\comp}{\mathcal{C}}
\newcommand{\opt}{\textsc{Opt}}

\newcommand{\alignedevenfrag}[1]{#1-aligned-even-fragmentation}

\newcommand{\trickyalignedevenfrag}[1]{selective #1-aligned-even-fragmentation}
\newcommand{\tricky}{selective }
\newcommand{\trickiness}{selectiveness }

\begin{document}

\title{Tight Bounds for Memory Allocation \\ With and Without Request Fragmentation} 

\author{%
\makebox[\textwidth][c]{%
\parbox[t]{0.31\textwidth}{\centering
{\large Michael A. Bender\par}
{\normalsize Stony Brook University\par}
{\small\texttt{bender@cs.stonybrook.edu}\par}}
\hfill
\parbox[t]{0.31\textwidth}{\centering
{\large Alex Conway\par}
{\normalsize Cornell University\par}
{\small\texttt{me@ajhconway.com}\par}}
\hfill
\parbox[t]{0.31\textwidth}{\centering
{\large Mart\'in Farach-Colton\par}
{\normalsize New York University\par}
{\small\texttt{martin@farach-colton.com}\par}}
}%

\par\vspace{0.5em}

\makebox[\textwidth][c]{%
\parbox[t]{0.31\textwidth}{\centering
{\large Hanna Koml\'os\par}
{\normalsize Max Planck Inst. for Informatics\par}
{\small\texttt{hkomlos@gmail.com}\par}}
\hfill
\parbox[t]{0.31\textwidth}{\centering
{\large William Kuszmaul\par}
{\normalsize Carnegie Mellon University\par}
{\small\texttt{kuszmaul.cmu.edu}\par}}
\hfill
\parbox[t]{0.31\textwidth}{\centering
{\large Nicole Wein\par}
{\normalsize University of Michigan\par}
{\small\texttt{nswein@umich.edu}\par}}
}%
}

\date{}
\maketitle

\begin{abstract}
The classical memory-allocation problem captures the task of placing objects of different sizes in memory, while minimizing the so-called memory high-water mark. It has been known since the early 1970s that the optimal competitive ratio for any deterministic online allocator is $\Theta(\log M)$, where $M$ is the volume high-water mark of the underlying request sequence.

This paper begins with a simple observation: many real-world allocators seem to bypass the 1971 lower bound by adopting a slightly different model for memory allocation. These allocators use what we call \emph{$k$-aggregate request fragmentation}, meaning that the memory allocator is permitted to break requests into multiple \emph{fragments}, so long as the all-time maximum number of simultaneous fragments is at most $k$ times the all-time maximum number of simultaneous requests. 

We consider the following basic question: Does request fragmentation fundamentally change the problem of memory allocation, and if so, how? Our results come with several surprises. Among these, we find that even using $k = 1 + o(1)$ request fragmentation, the optimal competitive ratio---which was $\Theta(\log M)$ in the classical setting---collapses to $\Theta(\log \log M)$. This result is shown to be tight with matching upper and lower bounds, applying to both deterministic and randomized algorithms. 

\end{abstract}
\pagenumbering{gobble}
%\newpage
\pagenumbering{arabic}

\section{Introduction}
Introduced in the early 1960s and 70s \cite{collins1961experience,maher1961problems,Robson71,Robson74,MR351411,Robson77}, the memory allocation problem\footnote{Also known as ``online memory allocation'', ``dynamic memory allocation'' and ``storage allocation''} is one of the oldest and most basic problems in computer science. 

In its classical form, a memory allocator must support two operations: a memory \defn{request}, which specifies some size $s > 0$, and requires an allocation of $s$ consecutive free slots in memory; and a memory \defn{free}, which specifies some previous request, and frees the slots formerly allocated to that request. Memory is represented as an unbounded one-dimensional array, and the goal of the allocator is to satisfy the given input sequence using \emph{as small a prefix of memory as possible}.

Formally, given an input sequence $W$ (also known as a \defn{workload}), and a memory allocation algorithm $\mathcal{A}$, the \defn{memory high-water mark} $\mathbf{\hwm(\alg, \workload)}$ is defined to be the smallest $q \ge 0$ such that all of $\alg$'s allocations on $\workload$ use slots in the interval $[1, q]$. The goal is to keep $\hwm(\alg, \workload)$ as close as possible to the so-called \defn{volume high-water mark} $\mathbf{\M(W)}$, which is the maximum, over all times $t$, of the sum of the sizes of the live requests at time $t$.\footnote{This quantity $\M(W)$ turns out to be equivalent, up to constant factors, to $\hwm(\opt, \workload)$, where $\opt$ is the \emph{offline} optimal allocation algorithm~\cite{MR2066646,MR968855,kierstead1991polynomial,DBLP:conf/mfcs/Slusarek89,approx_G96,gergov1999algorithms}.}

\paragraph{Classical results: bounds for both deterministic and randomized algorithms. }
An ideal algorithm would achieve \defn{competitive ratio} $\hwm(\alg, \workload) / \M(W)$ of $O(1)$ on all workloads $W$. However, in a classic 1971 paper, Robson \cite{Robson71} showed that no such (deterministic) algorithm exists. In fact, as a function of $M = M(W)$, the best competitive ratio achievable by any deterministic algorithm is $\Theta(\log M)$.\footnote{All logarithms in this paper are base 2.} This bound can be achieved by algorithms that do not know $M$, and is optimal even for algorithms that do know $M$ in advance. 

This discovery led researchers to explore two main directions, the task of optimizing the constants in Robson's $\Theta(\log M(W))$ bound~\cite{Robson74, MR351411, randomized_LNO96}, and the question of whether \emph{randomized} algorithms might be able to do better \cite{randomized_LNO96}. Interestingly, it has remained an open question whether randomized algorithms might be able to beat deterministic ones. The best lower bound, due to Luby, Naor and Orda in 1994~\cite{randomized_LNO96}, is a bound of $\Omega(\log M / \log \log M)$ which holds for the competitive ratio of any randomized algorithm, and which, again, holds regardless of whether or not the algorithm knows $M$ in advance.\footnote{The classical lower bounds are often stated in terms of other parameters, such as the request high-water mark or the ratio of the largest to the smallest request sizes. In their constructions and ours, all of these parameters differ by a constant factor.}

\paragraph{The modern approach: request fragmentation.} In practice, researchers and practitioners have been able to \emph{bypass} Robson's lower bound by allowing the allocation algorithm to perform \defn{request fragmentation}, in which the algorithm splits a request into multiple \defn{fragments} that are each allocated in different parts of memory.

Perhaps the most famous application of request fragmentation is at the operating-system level \cite{DBLP:books/wi/SGG2018,DBLP:books/daglib/0026450,arpaci-dusseau2018ostep}. When a program (or heap allocator) requests memory of some size $s$, modern operating systems break the request into fragments of different sizes (each fragment corresponds to a fixed-sized ``page'' or ``huge page'' in the operating system) and places these fragments non-contiguously in memory. The operating system then \emph{hides} this fragmentation from the user, by assigning the fragments virtual addresses that appear contiguous, even though the physical addresses are not. 

Given that modern allocators are capable of performing request fragmentation, one might wonder: Why not take this approach to the extreme? Why not, for example, break each request of size $s$ into $s$ different fragments each of size $1$? To see the problem with such an approach, it is helpful to once again consider the operating-system example. To support request fragmentation, the operating system must store a data structure called the \emph{page table} that maps the virtual address of each fragment to its physical address. Not only does the page table need to fit in RAM, it should ideally also fit in a relatively small hardware cache. But the size of the page table is directly proportional to the \emph{total number of fragments} in the system.

Thus, even if an allocator performs request fragmentation, it should seek to minimize the total number of fragments that are present at a time. With this goal in mind, we say that an allocator $\mathcal{A}$ uses \defn{$k$-aggregate fragmentation} if, on every workload $W$, the maximum number of fragments that are ever live at once is at most $k \cdot Q(W)$, where $Q(W)$ is the maximum number of \emph{requests} that are ever live at once. 

It is worth taking a few words to remark on why $k$-aggregate fragmentation is a natural notion to consider. This definition models the very practical capacity constraint of minimizing the size of the data structure needed to index fragments, which is determined by the maximum number of fragments ever in the system. In a non-fragmented allocator, the size required for the index would simply be the peak number of jobs ever alive in the system, so the ratio $k$ captures how the size complexity of the fragment index scales with respect to a) the input to the index, i.e., the number of jobs it must support, and b) an analogous unfragmented index. An ideal allocator should use $k$-aggregate fragmentation for some small $k$ (perhaps even $k = 1 + o(1)$). 

Although request fragmentation is widely used in practice, the question of how well it performs in theory has remained largely unexplored. How large does $k$ need to be in order for $k$-aggregate request fragmentation to break Robson's $\Omega(\log M)$ lower bound? And, if request fragmentation \emph{does} break Robson's lower bound, then what is the new best bound that one can achieve? 

\paragraph{This paper: Tight bounds for both classical and fragmented memory allocation.} The goal of this paper is to develop tight bounds for both deterministic and randomized memory allocation, with and without $k$-aggregate fragmentation. We summarize the results for the different regimes we study in Table~\ref{tab:summary}.

We give matching upper and lower bounds on the best competitive ratio that any allocation algorithm (deterministic or randomized) can achieve using $k$-aggregate fragmentation. Supposing $k \ge 2$, and given a known upper bound $\overline{M} \ge 1$ on $\M(W)$, we show that the optimal competitive ratio for any algorithm is
\begin{equation} 
\Theta(\log_k \log \overline{M}).
\label{eq:fragmented}
\end{equation}
Our upper bound is achieved by a simple deterministic algorithm that \emph{does not need to know $\overline{M}$ in advance}, while our lower bound applies to all randomized algorithms, including those that know $\overline{M}$ in advance. The lower bound is proven by combining the techniques which we use to prove \eqref{eq:classical} with another entirely separate set of arguments that capture the limitations of request fragmentation. 

We remark that the tradeoff given by \eqref{eq:fragmented} comes with several significant surprises. First, on the upper-bound side, it says that even a very \emph{small} amount of fragmentation ($k = O(1)$) is enough to trigger an exponential collapse in the competitive ratio, as a function of $\overline{M}$.\footnote{In most operating-system level applications, where memory is typically restricted to fit in, say, a $32$ GB RAM, and where memory is measured in $64$-bit words, one can trivially obtain a bound on $\overline{M}$ of, say $2^{32}$, which for $k = 3$ gives $\log_k \log \overline{M} \approx 3.1$.}
Second, on the lower-bound side, it says that using large values of $k$ comes with a comparably shallow payoff---increasing $k$ to be $\omega(1)$ divides the optimal competitive ratio by a relatively meager factor of $\Theta(\log k)$. 

We also generalize \eqref{eq:fragmented} in several ways. First, to the setting where $k = 1 + \epsilon$ for some $\epsilon \in (0, 1)$, where we find that the optimal competitive ratio becomes 
$$\Theta(\log \log \overline{M} + \log \epsilon^{-1}).$$
This means that, even using $k$ as small as $1 + 1 / \operatorname*{polylog} \overline{M}$ is enough to collapse the optimal competitive ratio to $\Theta(\log \log \overline{M})$. Second, we extend our results to the setting where $M$ is known to be in a restricted range $[\Mmin, \Mmax]$, and where $\Mmax / \Mmin = R$ for some parameter $R$.\footnote{Interestingly, our algorithms for this setting, given by Theorems \ref{thm:unknownMupper} and \ref{thm:unknownMuppereps}, actually require knowledge only of $\Mmin$, and not of $\Mmax$ or $R$.} In this case, we show that, for $k > 1$ and $R \ge 1$, the optimal competitive ratio becomes
\begin{equation}
\begin{cases}
    \Theta(1 + \log_k \log R) & \text{if } k \ge 2 \\
    \Theta(\log \log R + \log (k - 1)^{-1}) & \text{if }k \in (1, 2).
    \end{cases}
\label{eq:fragmented2}
\end{equation}
A notable special case is the setting where $\M(W)$ is fully known in advance (i.e., $R = 1$). In this case, for $k = \Theta(1)$, the optimal competitive ratio drops to $O(1)$. In contrast, in the classical setting, the optimal competitive ratio would remain at $\Omega(\log \M(W))$ (by \eqref{eq:classical}). 

Finally, we also consider a stronger notion of request fragmentation, called \defn{$k$-per-request fragmentation}, in which the allocator is permitted to break each \emph{individual} request into up to $k$ fragments, but cannot break any request into more than that.
Note that $k$-per-request fragmentation trivially cannot improve the optimal competitive ratio \eqref{eq:classical} by more than a factor of $k$, since any request of size $s$ must have a fragment of size at least $s/k$ -- and even if we allocated these fragments \emph{alone}, our memory-high water mark would be at least $\Omega(\frac{1}{k} M \log M)$, resulting in a competitive ratio of at least $\Omega(k^{-1} \log M)$. Our main result for $k$-per-request fragmentation is that, in fact, even this factor of $k$ is overly optimistic. We show that, using $k$-per-request fragmentation, the optimal competitive ratio (for both deterministic and randomized algorithms) is given by 
\begin{equation}
    \Theta(\log_{k + 1} M),
    \label{eq:fragmented3}
\end{equation}
regardless of whether or not any information about $M = \M(W)$ is known in advance. 

The contrast between Equations \eqref{eq:fragmented} and \eqref{eq:fragmented3} provides a strong theoretical justification for why, in practice, modern allocators (implicitly) use $k$-\emph{aggregate} fragmentation instead of $k$-\emph{per-request} fragmentation. The bounds achieved by the former are (at least) exponentially better than those achievable by the latter.

From our lower bound for $k$-per-request fragmentation, we also obtain as a corollary an improved lower bound for classical (non-fragmented) memory allocation. We prove that any randomized algorithm must incur a competitive ratio of at least 
\begin{equation} \Omega(\log M),
\label{eq:classical}
\end{equation}
even if $M = \M(W)$ is given to the algorithm in advance. Although this result closes the modest $\log\log M$ gap that remained after the randomized lower bound of Luby, et al. \cite{randomized_LNO96}, it is nevertheless an interesting byproduct of our analysis to settle the main conjecture in \cite{randomized_LNO96} and demonstrate that there is no asymptotic advantage to using randomization for classical memory allocation.

Finally, in \Cref{app:average}, we also consider several other notions of fragmentation that, a priori, might seem different, and we prove reductions showing that these are essentially equivalent to per-request fragmentation.

Together, our results paint, for the first time, a complete picture of what the landscape of memory-allocation algorithms look like, both for deterministic and randomized algorithms, with and without memory fragmentation. Our results reveal that the use of $k$-aggregate fragmentation, even for very small values of $k$, results in a near total collapse of the classical lower bounds -- thereby justifying why real-world allocators have converged to this solution in practice. Our results further demonstrate that different types of fragmentation ($k$-aggregate vs $k$-per-request) result in very different tradeoff curves, meaning that almost all of the algorithmic value to be obtained from fragmentation requires different requests to be broken into different amounts of fragments.  Finally, our results reveal that randomization, despite offering an advantage in other related settings such as paging \cite{SleatorTa85,fiat1991competitive}, offers no asymptotic advantage for memory allocation.

\begin{table}
    \centering
    \begin{tabular}{|c|c|c|c|c|}
        \hline
        Fragmentation & Parameter Range & Bound Type & Competitive Ratio & Reference\\
        \hline
        \hline
         None & N/A & Upper & $O(\log M)$ & \cite{Robson77} \\
        \hline
         None & N/A & Lower & $\Omega(\log M)$ & \Cref{cor:main-no-fragments} \\
         \hline
        $k$-per request & $k \geq 1$ & Upper & $O(\log_k M)$ & \Cref{clm:upper-per} \\
        \hline
         $k$-per request & $k \geq 1$ &Lower& $\Omega(\log_k M)$ & \Cref{cor:old-main-lower} \\
         \hline
        $k$-aggregate & $k \geq 2$ & Upper & $O(\log_k\log(\overline{M}))$ & \Cref{thm:unknownMupper-paraphrased} \\
        \hline
         $k$-aggregate & $k \geq 2$ & Lower & $\Omega(\log_k\log (\overline{M}))$ & \Cref{thm:ag} \\
         \hline
        $(1+\epsilon)$-aggregate & $\epsilon \in (0,1)$ & Upper & $O(\log\log\overline{M} + \log\epsilon^{-1}))$ & \Cref{thm:unknownMuppereps-paraphrased} \\
        \hline
         $(1+\epsilon)$-aggregate & $\epsilon \in (0,1)$ & Lower & $\Omega(\log\log \overline{M} + \log\epsilon^{-1}))$ & \Cref{cor:epslower-final} \\
         \hline
    \end{tabular}
    \caption{Summary of Results for Different Types of Fragmentation. Bounds are given as competitive ratios to an offline optimal algorithm, which has memory high-water mark $\Theta(M)$, where $M$ is the volume high-water mark (the maximum number of simultaneously occupied memory slots at any time) of the workload. For $k$-aggregate fragmentation, $\overline{M}$ denotes an upper bound on $M$.}
    \label{tab:summary}
\end{table}

\paragraph{Related work. }
The problem of dynamic storage allocation first came to prominence in the academic literature in 1961, when \textit{Communications of the ACM} published a special issue focusing exclusively on the topic~\cite{cacm-1961-10}. In that issue Collins initiated the experimental study of several algorithms that remain popular today, including First Fit, Best Fit, Worst Fit, and Random Fit~\cite{collins1961experience}. Another popular algorithm, the buddy allocator was invented in 1963 by Markowitz and first described by Knowlton in 1965~\cite{Knowlton65,10.1145/365758.365792}.
In 1968, Knuth's \textit{The Art of Computer Programming}, Vol. 1, dedicated a section to Dynamic Storage Allocation, fully codifying the problem and standardizing much of the terminology~\cite{Knuth97}.

The modern formulation of the problem, in terms of competitive ratios, was introduced by Robson in 1971~\cite{Robson71,Robson74,Robson77}, and  independently by Woodall in 1974 \cite{MR351411}. Both authors \cite{Robson71, MR351411} proved matching $\Theta(\log M)$ upper and lower bounds for the best competitive ratio of any deterministic algorithm. Robson \cite{Robson77} further determined sharp worst-case bounds for the competitive ratios achieved by the Best Fit and First Fit algorithm. His results came with several concrete lessons that continue to impact practice today. Notably, the Best Fit algorithm, which may intuitively seem like a good heuristic, can incur a competitive ratio as bad as $\Omega(M)$; while the First-Fit algorithm, which might intuitively seem like a worse algorithm, achieves an asymptotically optimal competitive ratio of $O(\log M)$.

In the 1980s and 90s, it became increasingly clear that, for many online problems (e.g., the closely related online paging problem  \cite{SleatorTa85,fiat1991competitive}), randomization can be used to bypass barriers that deterministic algorithms face. Luby, Naor, and Orda raised the question of whether randomized allocation algorithms might achieve a better competitive ratio than $\Theta(\log M)$ \cite{randomized_LNO96}. The authors \cite{randomized_LNO96} were able to prove a lower bound of $\Omega(\log M / \log \log M)$, and conjectured the stronger lower bound of $\Omega(\log M)$ proven as one of the results in this paper.

Dynamic storage allocation is closely related to several other combinatorial scheduling problems.
In 1983, Coffman connected dynamic storage allocation to dynamic bin packing~\cite{CoffmanGaJo83,DBLP:journals/jal/BakerC82,DBLP:journals/siamcomp/CoffmanGJ83}, and there has been a line of research in this direction~\cite{packing_CS88,CoffmanGaJo83,CoffmanJoSh93,CoffmanJoSh97,CoffmanGaJo96,GalambosWo95}.
This connection can be explicitly seen through a ``rectangle view,'' where one axis is memory address and the other is time~\cite{DBLP:journals/ita/ChrobakS88}.
Surprisingly, recent work suggests that bin-packing algorithms may even yield good allocators in practice~\cite{DBLP:conf/iwmm/LamprakosXCS23}.
There is also a strong connection to online interval coloring~\cite{DBLP:conf/cocoon/Narayanaswamy04}, which has influenced several of the results on dynamic storage allocation.

Another version of the problem that has been studied is the known-duration version, where the algorithm is told each request's duration on arrival \cite{NaorOrPe00, durations_KP00}. Perhaps surprisingly, this version of the problem \emph{continues} to be subject to a $\tilde{\Omega}(\log M)$ lower bounds on the competitive ratio \cite{durations_KP00}.

The offline version of the problem has also received a great deal of attention \cite{10.5555/574848,slusarek1987npstorage,MR2066646,MR968855,kierstead1991polynomial,DBLP:conf/mfcs/Slusarek89,approx_G96,gergov1999algorithms}.
It is known to be NP-Complete~\cite{10.5555/574848,slusarek1987npstorage}, and much of the research has focused on achieving increasingly better approximation ratios.
Early work by Kierstead and Slusarek relied on the connection to interval coloring~\cite{MR968855,kierstead1991polynomial,DBLP:conf/mfcs/Slusarek89}, but Gergov showed that better approximation ratios than those implied by that problem are possible~\cite{approx_G96,gergov1999algorithms}, including the best-known 3-approximation.

Finally, researchers have also studied many other variations of the problem, such as dynamic storage \emph{re-}allocation~\cite{HallPo04,SandersSiSk09,UnalUzKi97,BenderFaFe13,reallocation_BFFFG14,farach2024nearly,lim2015dynamic,bender2015cost,kuszmaul2023strongly,farach2024nearly,jin_realloc}, allocation with restrictions on request sizes~\cite{MR2066646, approx_G96, approx_LLQ04, murthy1999approximation}, and allocation under stochastic arrivals and departures~\cite{MR990048}.

\section{Preliminaries}

In this section, we present basic definitions and notations that will be used throughout the rest of the paper. The reader will note that several of these definitions overlap with those already discussed in the introduction.

Formally, the online memory allocation problem is defined as follows. The input is a \defn{workload}, consisting of an arbitrary sequence of \defn{memory requests} and \defn{memory frees} (or \defn{requests} and \defn{frees} for short).  Memory is represented as an unbounded 1-dimensional array.
Each memory request $r$ specifies a positive integer size $|r|$, signifying the number of memory slots that must be reserved to satisfy the request.

The allocation algorithm must decide which memory slots to allocate to satisfy a request before seeing the next request. These reserved slots are called the \defn{allocation} of the request. Each allocation must occupy a contiguous block of memory, and no two  allocations can overlap.
When a request arrives, it is said to be \defn{live}. 
A free specifies a live memory request whose allocation is consequently \defn{deallocated}, after which the memory slots in the allocation are available for future allocations.

For a memory allocation algorithm $\mathcal{A}$ and workload $W$, the \defn{memory high-water mark} $\mathbf{\hwm(\alg, \workload)}$ of algorithm $\alg$ on workload $\workload$ is the highest index in the memory array that $\alg$ allocates on $\workload$; and the \defn{volume high-water mark}, $\mathbf{\M(W)}$, of $\workload$ is the maximum number of memory slots ever simultaneously live in $\workload$. The (expected) \defn{competitive ratio} of the algorithm on the workload is defined by 
$$C(\alg, \workload) = \E\left[\frac{\hwm(\alg, \workload)}{\M(W)}\right],$$
where the expectation is taken over any randomness the algorithm may use. The goal of an allocation algorithm is to achieve as small competitive ratio as possible. In addition to these quantities, it will often also be helpful to refer to the \defn{request high-water mark} $Q(\workload)$, which is the maximum number of requests that are ever simultaneously live in workload $\workload$.

Finally, we will also investigate allocation algorithms that are permitted to use either $k$-aggregate or $k$-per-request fragmentation, for some parameter $k > 1$. These algorithms break each request $r$ into one or more fragments whose (positive integer) sizes sum to the size of $r$, and then allocate each fragment $q$ to some $|q|$ consecutive slots of free memory---these allocations are considered to be live until the request $r$ is freed.  For a given request $r$ broken into some number $f \ge 1$ of fragments, the first of these fragments is said to be a \defn{complete fragment}, and the others are said to be \defn{partial fragments}.

In the setting of \defn{$k$-aggregate fragmentation}, which is well-defined for any real $k \ge 1$, the algorithm must ensure, for any workload $W$ that there are never more than $k\cdot Q(W)$ fragments alive at a time. When $k = 1 + \epsilon$ for some $\epsilon$, this is equivalent to guaranteeing for every workload $W$ that there are never more than $\epsilon \cdot Q(W)$ live \emph{partial} fragments at a time. In the setting of $k$-per-request fragmentation, which is defined only for positive integer $k$, the algorithm must break each fragment into at most $k$ fragments. 

In addition to considering algorithms that use $k$-aggregate fragmentation (resp.~$k$-per-request fragmentation) our lower bounds will also apply to what we call \defn{expected-$k$ aggregate fragmentation} (resp.~\defn{expected-$k$ per-request fragmentation}), where the algorithm must limit the total number of live fragments at a time (resp.~number of fragments per request) to $k\cdot Q(W)$ (resp.~$k$) \emph{in expectation}.

\section{Technical Overview}

In this section, we present an overview of the technical ideas in the paper, along with how the various techniques and lower bounds fit together. The goal of the section is to arm the reader with the high-level ideas needed to be able to easily understand the full versions of the arguments later in the paper. To streamline the exposition, we will typically focus either on the setting where there is no fragmentation, or on the most basic setting for $k$-aggregate fragmentation, which is $k = 2$.

\subsection{Warmup: A Tight Lower Bound Without Fragmentation}

We begin by discussing how to prove an $\Omega(\log M)$ lower bound for classical (fragmentation-free) memory allocation. In addition to resolving an open question due to Luby, Naor, and Orda~\cite{randomized_LNO96}, the high-level technical approach ends up being essential in all of our other, more general, lower bounds.

\begin{theorem}[Special case of Theorem \ref{thm:main-lower}]
For any $M > 0$, there exists a workload $\adv$, with volume high-water mark $M(\adv) = O(M)$, that forces any randomized online memory allocation $\mathcal{A}$ to incur an expected competitive ratio of at least $\Omega(\log M)$, or equivalently, an expected memory high-water mark $\E[S(\mathcal{A}, \adv)]$ of at least $\Omega(M \log M)$.
\label{thm:classical}
\end{theorem}

It is worth taking a moment to contrast our lower bound construction with that of Luby et al. Both of our constructions use an adversarial workload in rounds where jobs increase geometrically in size, with randomized deletions of jobs from prior rounds. Luby et al. prove their lower bound using a sparse witness set on surviving jobs and their consecutive distances, and it is the sparsity of the witness set that contributes the extra $\log\log M$ factor. Our approach is fundamentally different. We partition memory into aligned diadic blocks and show via a potential function argument that, on average, blocks can receive only a limited total volume of allocations across all rounds of the workload. Thus, we conclude that many memory blocks must be used in order to serve the request sequence. Additionally, while Luby et al. delete items with probability $1/2$, our adversarial workload relies on deleting items with a smaller probability to create a larger population of small surviving items, which block larger allocations.

\paragraph{The workload. } The workload $\adv$ that we use for Theorem \ref{thm:classical} is itself very simple. The workload $\adv$ consists of rounds $r = 0, 1, 2, \ldots, \log M$, where each round $r$ proceeds as follows. In round $r$, the workload $\adv$ first frees a uniformly random $0.1$ fraction of the remaining live requests in each size class; and then makes $M/2^r$ new requests of size $2^{r}$. This is the entire workload.

It is worth noting why $\adv$ frees requests at the rate that it does. The workload is designed so that, at the end of round $r$, almost all of the live requests are still from much earlier rounds (there are roughly $(M/2^r) \cdot 2^i / 1.1^i$ live requests from round $r - i$); but so that almost all of the \emph{volume} from live requests is due to recent rounds (the volume of requests from round $r - i$ is roughly $M / 1.1^i$). 

This \defn{volume vs.~count dichotomy}, in which earlier rounds contribute very little volume but contribute almost all of the live requests, will be essential to our probabilistic analysis. Intuitively, the dichotomy makes it very difficult for the algorithm $\mathcal{A}$ to allocate requests in earlier rounds $r - i$ without those requests (many of which will still be present by round $r$) causing problems for the round-$r$ allocations. 

\paragraph{High-level approach.}
Define the \defn{round-$r$ blocks} in memory, denoted $\mathcal{B}_r$, to be the memory-aligned blocks of size $2^r$ in the memory array (i.e., the disjoint blocks of size $2^r$ whose starting location is a multiple of $2^r$). We may assume (WLOG) that, when a size-$2^r$ request is made, it is allocated to a block in $\mathcal{B}_r$ (i.e., allocations are power-of-two memory-aligned).

For a given round $r$ and a block $B \in \mathcal{B}_r$, define $I_B$ to be the indicator random variable for whether a size-$|B|$ request is ever allocated in that block (it counts even if the request is later deallocated). 

Define $\psi^*(B) = I_B$ if $r = 0$ and 
\[\psi^*(B) = \psi^*(B_1) + \psi^*(B_2) + |B| \cdot I_B\]
if $r > 1$, where $B_1$ and $B_2$ are the two level-$(r - 1)$ blocks contained in $B$. In other words, $\psi^*(B)$ is the total volume of allocations made into $B$ in all rounds up to $r$.

The total amount of volume ever allocated (in any block at any time) is given by
$$\sum_{B \in \mathcal{B}_{\log M}} \psi^*(B) = M (1 + \log M).$$ 
Imagine that we could show that all blocks $B \in \mathcal{B}_{\log M}$ satisfied $\psi^*(B) \le O(2^{\log M}) = O(M)$. This is to say that each block $B \in \mathcal{B}_{\log M}$ is only able to serve a \emph{total} volume of requests roughly proportional to $|B|$. Then, since the total volume of all requests is $M (1 + \log M)$, we could conclude that the number of blocks $B \in \mathcal{B}_{\log M}$ that receive at least one allocation must be quite large, at least $\Omega(\log M).$ Since each such block has size $M$, this would imply a memory high-water mark of $\Omega(M \log M)$, as desired. 

Rather than trying to directly show that all blocks $B \in \mathcal{B}_{\log M}$ satisfy $\psi^*(B) \le O(2^{\log M}) = M$ (or $\E[\psi^*(B)] \le O(2^{\log M})$), we will settle for a weaker but nonetheless sufficient claim. Namely, that, for some large constant $c$,
\begin{equation}\E\left[\sum_{B \in \mathcal{B}_{\log M}} \min(c M, \psi^*(B))\right] \ge \Omega(M \log M).
\label{eq:foo1}
\end{equation}
One can think of \eqref{eq:foo1} as saying that (in expectation) a constant fraction of all volume ever allocated is allocated to a block $B \in \mathcal{B}_{\log M}$ that, \emph{at that point in time}, had received less than $O(M)$ total volume of allocations. 

Moreover, rather than analyze the quantity in \eqref{eq:foo1} directly, we instead analyze a quantity $\psi(B)$ defined as follows. For each round $r$, and for each block $B\in\mathcal{B}_r$, define 
\[\psi(B) = \min(c|B|, \psi(B_1) + \psi(B_2) + |B| \cdot I_B),\]
where as above, $B_1$ and $B_2$ are the two level-$(r - 1)$ blocks contained in $B$ (or where $\psi(B_1) + \psi(B_2)$ is evaluated as $0$, if $r = 0$). Also define the \defn{rounding difference} $\delta(B)$ for block $B$ to be the difference $\psi(B_1) + \psi(B_2) - \psi(B)$. Note that this rounding difference $\delta(B)$ can only be non-zero if $I_B = 1$, and is guaranteed to be at most $|B|$.

To prove \eqref{eq:foo1}, it suffices to show that
\begin{equation}\E\left[\sum_{B \in \mathcal{B}_{\log M}} \psi(B)\right] \ge \Omega(M \log M).
\label{eq:foo2}
\end{equation}
But notice that, in general, 
$$\sum_{B \in \mathcal{B}_{\log M}} \psi^*(B) -  \sum_{B \in \mathcal{B}_{\log M}} \psi(B) = \sum_{r \le \log M} \sum_{B \in \mathcal{B}_r} \delta(B).$$
Since the first sum is exactly $M (1 + \log M)$, the task of proving \eqref{eq:foo2} is equivalent to showing that 
\begin{equation} \E\left[\sum_{r \le \log M} \sum_{B \in \mathcal{B}_r} \delta(B)\right] \le (1 - \Omega(1)) M \log M.
\label{eq:differencesbound}
\end{equation}
Thus, the entire lower bound actually comes down to the task of \emph{upper bounding} the sum of the rounding differences $\delta(B)$. This sum is trivially at most $M \cdot (1 + \log M)$, and in order for our bound to be useful, we just need to achieve (any) bound that is a constant factor smaller than this. If we can prove such a bound, then we will have a proof of Theorem \ref{thm:classical}.

As we will see, the task of bounding (sums of) rounding differences is amenable to a remarkably clean charging argument. This argument is also the part of the proof that exploits the ``volume vs.~quantity dichotomy'' discussed earlier.

\paragraph{Bounding the rounding differences. }To upper bound 
\begin{equation}
    \sum_{r \le \log M} \sum_{B \in \mathcal{B}_r} \delta(B), \label{eq:differences}
\end{equation}
we perform what is essentially a charging argument. Suppose some round-$r$ block $B \in \mathcal{R}$ has non-zero rounding difference $\delta(B) > 0$. Then, the total amount of volume $\psi^*(B)$ allocated to $B$ in the first $r$ rounds (including volume freed) must be at least $c |B|$. Writing $|B| = \Theta\left(\sum_i\frac{|B|}{2^i}\cdot\frac{2^i}{1.1^i}\right)$, we have that there must be some round $r - i$ that contributed at least
\begin{equation} \Omega(2^i / 1.1^i)
\label{eq:manyallocations}
\end{equation}
allocations of size $|B| / 2^i$ to $B$. In order for $\delta(B)$ to be non-zero, there must be a size-$|B|$ allocation that is made to $B$ in round $r$. But such an allocation is only possible if all of the earlier round-$(r - i)$ allocations to $B$ were freed during rounds $r - i + 1, \ldots, r$. In this case, we say that $B$ \defn{charges} its rounding difference, which is at most $|B| = 2^r$, to round $r - i$.

For a given pair of rounds $k$ and $r = k + i$, define $\mathcal{B}_{k, i}$ to be the set of blocks $B \in \mathcal{B}_{k + i}$ with the property that $B$ received at least $\Omega(2^i / 1.1^i)$ allocations in round $r + i$, and all of those allocations were freed in rounds $r + i + 1, \ldots, r$. The set $\mathcal{B}_{k + i}$ contains (among other blocks) all round-$r$ blocks $B$ that charge their rounding difference to round $k$. To upper bound \eqref{eq:differences}, it suffices instead to bound
\begin{equation}\sum_{0 \le k \le \log M + 1} \sum_{1 \le i \le \log M + 1 - k} |\mathcal{B}_{k, i}| \cdot 2^{k + i}.
\label{eq:Bki}
\end{equation}

Here is where we finally benefit from the ``volume vs.~quantity dichotomy'' discussed earlier. Fix some $k, i$ and $r = k + i$, and consider a block $B \in \mathcal{B}_r$ that receives at least $\Omega(2^i / 1.1^i)$ allocations in round $k$. The expected number of those requests that remain live by the end of round $r$ is $\Omega(2^i / (1.1)^{2i})$, which is still quite large! By a Chernoff bound (for negatively associated random variables), we can conclude that
\begin{equation}\Pr[B \in \mathcal{B}_{k, i}] \le 2^{-\Omega(2^i / (1.1)^{2i})} \le 2^{-\Omega(1.5^i)}.
\label{eq:PrB}
\end{equation}
On the other hand, the total number of blocks $B \in \mathcal{B}_r$ that receive $\Omega(2^i / 1.1^i)$ allocations in round $k$ is itself, at most $(M / 2^k) / \Omega(2^i / 1.1^i) = O(M \cdot 1.1^i / 2^r)$. Combining this with \eqref{eq:PrB} allows us to deduce that
\begin{equation}\E[|\mathcal{B}_{k, i}| \cdot 2^r] \le M \cdot 2^{-\Omega(1.5^i)}.
\label{eq:prec}
\end{equation}
Here is where the magnitude of the constant $c$, used in the definition of $\psi$ comes into play. If $c$ is chosen large enough, then we can replace the asymptotic notation in \eqref{eq:prec} with an explicit bound of, say,
\begin{equation}\E[|\mathcal{B}_{k, i}| \cdot 2^r] \le M \cdot 2^{-1.5^i} / 100.
\label{eq:afterc}
\end{equation}
Plugging this into \eqref{eq:Bki} and simplifying, we can conclude that \eqref{eq:Bki} has expected value at most, say, $(M \log M) / 2$. This implies the same bound for the sum \eqref{eq:differences} of the rounding differences, which implies the entire chain of \eqref{eq:differencesbound}, \eqref{eq:foo2}, and \eqref{eq:foo1}. Having concluded \eqref{eq:foo1}, we can then complete the proof of Theorem \ref{thm:classical}.

\subsection{Memory Allocation with Request Fragmentation} 

Next we discuss the more difficult task of proving lower bounds for algorithms that are permitted to perform (various types of) request fragmentation. 

\paragraph{Extending to per-request fragmentation.} We begin by extending our techniques in the previous section to prove the following theorem:

\begin{theorem}[Later restated as Theorem \ref{thm:main-lower}]
Consider $k \ge 2$, let $Q \le M \in \mathbb{N}$, and let $\mathcal{A}$ be a randomized memory allocation algorithm using $k$-expected per-request fragmentation. Then there exists a workload $V(M, Q)$ that has at most $Q$ live requests at a time, each with a size in $[M / Q, M]$; that has volume high-water mark $M(V(M, Q)) = \Theta(M)$; and that forces $\mathcal{A}$ to incur a competitive ratio of $\Omega(\log_k Q)$.
\label{thm:strongextension}
\end{theorem}

In addition to being a statement about per-request fragmentation, Theorem \ref{thm:strongextension} also ends up serving as the main subroutine in our final lower bound (for $k$-aggregate fragmented algorithms) which we discuss in a moment. 

The proof of Theorem \ref{thm:strongextension} is the most intricate argument in the paper. It consists roughly of two pieces: (1) a series of reductions, simplifying (in a WLOG fashion) the types of behavior that $\mathcal{A}$ might exhibit; and (2) an extension of the argument from the previous section to apply to algorithms that perform fragmentation. The reductions are deferred to \Cref{app:reductions} and the full details of the proof are shown in \Cref{sec:lower-per-request}.

\paragraph{Lower bounds for $k$-aggregate fragmentation.}
Finally, we can now discuss how to prove a bound of $\Omega(\log_k \log \overline{M})$ on the competitive ratio of any memory-allocation algorithm using $k$-aggregate fragmentation on workloads with volume high-water mark bounded above by $\overline{M}$. 

\begin{theorem}[Special case of Theorem \ref{thm:fragmento}]
Let $k \ge 2$ and $\overline{M} \in \mathbb{N}$ and let $\mathcal{A}$ be a randomized online memory-allocation algorithm that performs $k$-aggregate request fragmentation. Then, there exists a workload $W$ with volume high-water mark $M(W) \in [1, \overline{M}]$ on which $\mathcal{A}$ incurs a competitive ratio of at least $\Omega(\log_k \log \overline{M})$. 
\label{thm:ag}
\end{theorem}

To describe the proof of Theorem \ref{thm:ag} as simply as possible, let us assume that $k = 2$. Then the workload used to prove Theorem \ref{thm:ag} can be described as follows. Let $Q = \sqrt{\log \overline{M}}$.\footnote{The proof described here differs slightly from the proof of Theorem \ref{thm:fragmento} in how the parameters are set up. Specifically, Theorem \ref{thm:fragmento} imposes both upper and lower bounds $\Mmax$ and $\Mmin$ on the volume high-water mark, and makes the (WLOG) assumption that $\Mmax/\Mmin \ge \log (\Mmax/\Mmin)$. With this slightly more complicated setup, one can replace $Q$ as described here directly with $\log (\Mmax / \Mmin)$, as in the proof of Theorem \ref{thm:fragmento}.} The workload $W$ consists of rounds $1, 2, \ldots, \log \frac{\overline{M}}{Q}$, where each round $r$ is implemented as follows:
\begin{enumerate}
    \item Begin to perform the workload $V_r := V(2^r \cdot Q, Q)$ given by Theorem \ref{thm:strongextension}.
    \item Call a request $q \in V_r$ \defn{heavily fragmented} if the \emph{expected} number of fragments that $\mathcal{A}$ breaks $q$ into is greater than $2k$.
    \begin{enumerate}
        \item If there exists a heavily fragmented request $q_r \in V_r$, then end the workload $V_r$ after $q_r$ is requested, and free all live requests in $V_r$ except for $q_r$. 
        \item If there does not exist a heavily fragmented request $q_r \in V_r$, then complete the workload $V_r$ and terminate the construction. (Do not proceed to phase $r + 1$.)
\end{enumerate}
\end{enumerate}

Notice that the description of the workload $W$ depends, in part, on the algorithm $\mathcal{A}$ being used. It does \emph{not}, however, require us to know anything about the random bits used by $\mathcal{A}$. Critically, we can determine in each phase $r$ whether there exists a heavily fragmented request $q_r \in V_r$ based only on the specification of $\mathcal{A}$, because the question of whether $q_r$ is heavily fragmented depends only the \emph{expected} number of fragments that $q_r$ is broken into, rather than the number that is actually realized.

One way for $W$ to terminate is if we reach a phase $r$ in which there is no heavily fragmented request. In this case, we know by Theorem \ref{thm:strongextension} that $\mathcal{A}$ incurs an expected memory-high water mark $\E[S(\mathcal{A}, V)]$ of at least
$$M(V_r) \log Q = 2^r Q \log Q.$$
On the other hand, the volume high-water mark $M(V)$ is at most
\begin{equation} \sum_{i = 1}^r M(V_i) \le \sum_{i = 1}^r O(2^i \cdot Q) = O(2^r Q).
\label{eq:boundingvol}
\end{equation}
It follows, in this case, that $\mathcal{A}$ incurs a competitive ratio of at least $\Omega(\log_k Q) = \Omega(\log_k \log \overline{M})$. 

The other way for $W$ to terminate is after all $L := \log \frac{\overline{M}}{Q}$ phases complete. In this case, we are left with $L$ live requests $q_1, q_2, \ldots, q_{L}$, each of which is broken into more than $2k$ expected fragments. The total expected number of fragments is therefore greater than $2Lk$. 

On the other hand, the total number of requests $Q(V)$ that are ever simultaneously live in $W$ is at most
$$L + Q \le 2L,$$
where the $L$ term considers the requests $q_1, q_2, \ldots, q_L$, and where the $Q$ term considers the live requests within a given workload $V_r$ during a given phase $r$. Because $\mathcal{A}$ guarantees at most $k$-aggregate request fragmentation, it follows that the expected number of fragments can be \emph{at most} $2Lk$, a contradiction. 

In \Cref{sec:agg}, we also show the following simple corollary, which extends the above lower bounds to the $(1+\epsilon)$-aggregate fragmentation regime for $\epsilon \in (0,1)$.

\begin{restatable}{corollary}{epslower}
Let $\epsilon \in (0, 1)$ and $\overline{M} \ge \epsilon^{-1}$. 
Let $\mathcal{A}$ be a randomized online memory-allocation algorithm that performs $(1+\epsilon)$-aggregate request fragmentation. Then, there exists a workload $W$ with volume high-water mark $M(W) \in [1, \overline{M}]$ on which $\mathcal{A}$ incurs a competitive ratio of at least $\Omega(\log \log \overline{M}+\log\epsilon^{-1})$. 
\label{cor:epslower-final}
\end{restatable}

\paragraph{Understanding why the lower bound produces $\Omega(\log_k \log \overline{M})$.}
At this point, it is worth taking a moment to think about the structure of the above construction, in order to understand why $\Omega(\log_k \log \overline{M})$ is the bound that naturally comes out.

The first question to consider is why we must require each round $r$ to use $V_r = V(2^r \cdot Q, Q)$, as opposed to using $V(Q, Q)$ for every round. This is so that \eqref{eq:boundingvol} holds---that is, so that, if the algorithm terminates in some phase $r$, its volume high-water mark is dominated by the volume high-water mark of $V_r$. This requirement may seem minor, but it is the reason that we must cap the \emph{number} of phases to $L = O(\log \overline{M})$.

The next question to consider is why we must set $Q$ to be so small, and, in particular, why $Q \le L$. The main bottleneck to making $Q$ large is that, in each phase $r$, the number of simultaneously live requests may be up to $Q$. But, at the end of the construction, we have $L$ requests $q_1, q_2, \ldots, q_L$ that are each broken into more than $2k$ expected fragments. In order for this to lead to a contradiction (i.e., we have too many fragments), we need $Q \le L$. 

When we put these two constraints together, we get that $L = O(\log \overline{M})$ and $Q \le O(L)$. The competitive ratio that we recover is given by the performance of Theorem \ref{thm:strongextension} on a workload of the form $V(2^i \cdot Q, Q)$, which results in a competitive ratio of $\Omega(\log_k Q) = \Omega(\log_k \log \overline{M})$. This is where the bound of $\log_k \log \overline{M}$ ultimately comes from. 

\paragraph{A matching upper bound. } \emph{A priori}, the above lower bound may seem quite loose. 
Nonetheless, there is also a remarkably simple algorithm that achieves a matching upper bound.\footnote{Interestingly, although the algorithm in Section \ref{sec:aggupper} is quite simple, it was not at all obvious to the current authors until \emph{after} the $\Omega(\log_k \log \overline{M})$ lower bound was discovered. The algorithm was then obtained by trying to design an algorithm for which the lower bound would be tight.} 
In \Cref{sec:aggupper}, we prove the following.

\begin{theorem}[Special case of \Cref{thm:unknownMupper}]
For $k\geq 2$ and $\overline{M} \geq 0$, there is a (deterministic) online allocation algorithm using $k$-aggregate fragmentation that achieves competitive ratio $O(\log_k \log \overline{M})$ on every workload $W$ with volume high-water mark $M(W) \leq \overline{M}$. 
\label{thm:unknownMupper-paraphrased}
\end{theorem}

The entire algorithm (again, for the case of $k = 2$) is as follows. At all times $t$, keep track of both the volume high-water mark $\Mguess(t)$ and the request high-water mark $\Qguess(t)$ so far. Then, when processing a request of size $s$, break it into 
$$\left\lceil \frac{s \Qguess(t)}{c \Mguess(t) (1 + \log^2 \Mguess(t))}\right\rceil$$
fragments of size at most $c \Mguess(t) (1 + \log^2 \Mguess(t)) / \Qguess(t)$ each, where $c$ is some appropriately chosen positive constant. Finally, allocate each of these fragments using a first-fit policy, placing the fragment in the free region of memory large enough to host the fragment.\footnote{The reader may notice that the above algorithm description is in a few (it turns out, purely aesthetic) ways slightly different from the one in Section \ref{sec:aggupper}. These differences, and specifically the distinction between having implicit phases (as in this algorithm) versus explicit phases (as in Section \ref{sec:aggupper}) is only to streamline the exposition in Section \ref{sec:aggupper} (allowing us, in particular, to prove Theorem \ref{thm:unknownMupper} using Theorem \ref{thm:knownMupper} directly as a subroutine).}

Although the description of the above algorithm is quite simple, the analysis is perhaps surprisingly subtle, and is deferred to Section \ref{sec:aggupper}. 
In that section, we also generalize this algorithm to the $(1+\epsilon)$-aggregate fragmentation regime for $\epsilon\in(0,1)$, resulting the following upper bound matching our lower bound.

\begin{theorem}[Special case of \Cref{thm:unknownMuppereps}]
For $\epsilon \in (0, 1)$ and $\overline{M} \geq 0$, there is a (deterministic) online allocation algorithm $\mathcal{A}$ using $(1 + \epsilon)$-aggregate fragmentation that achieves competitive ratio $O(\log \log \overline{M} + \log \epsilon^{-1})$ on every workload $W$ such that $M(W) \leq \overline{M}$.
\label{thm:unknownMuppereps-paraphrased}
\end{theorem}

\section{A Tight Lower Bound for Expected-$k$ Per-Request Fragmentation}\label{sec:lower-per-request}

In this section, we prove the following theorem:

\begin{theorem}\label{thm:main-lower}
    For any $\advmem$ and $\advreq \leq \advmem$, there exists a randomized workload $\workload$ with volume high-water mark $\M(\workload) = \advmem$ and $\Q(\workload) = \advreq$, such that for any (possibly randomized) allocation algorithm $\alg$ with $k$-expected per-request fragmentation, the competitive ratio $\comp(\alg,\workload)$ of $\alg$ on $\workload$ satisfies
    \[
        \comp(\alg,\workload) \ge \Omega \left(\log_{\max\{k,2\}}\advreq\right).
    \]
\end{theorem}

We immediately get the following corollaries.
First, setting $\Q=\M$ yields the following lower bound for the competitive ratio in terms of $\M$. 

\begin{corollary}\label{cor:old-main-lower}
    Given $M > 0$ and a (possibly randomized) allocation algorithm $\alg{}$ with $k$-expected per-request fragmentation, there exists a workload $W$ with volume high-water mark $M(W) = M$ such that
    \[ \comp(\alg,W) \geq \Omega\left(\log_{\max\{k,2\}}\M\right), \]
\end{corollary}

It should be noted that \Cref{cor:old-main-lower} is tight, as can be seen with the following simple argument.

\begin{theorem}\label{clm:upper-per}
For any $k \ge 1$, there exists a deterministic allocation algorithm using $k$-per-request fragmentation that, on any workload $W$, achieves competitive ratio 
 \[ \comp(\alg,W) = O\left(\log_{\max\{k,2\}}\M(W)\right). \]
\end{theorem}

\begin{proof}
    For $k = 1$, this follows from the classical $O(\log M(W))$ upper bound for the first-fit algorithm \cite{Robson71,Robson77}. For $k \ge 2$, we can use $k$-per-request fragmentation to break each request into fragments whose sizes are within an $O(1)$ factor of a power of $k$. Up to constant factors, we may therefore assume that each fragment has a size exactly equal to a power of $k$. This reduces the number of distinct allocation sizes to at most $\log_k M(W)$. The claimed result then follows from the fact that, on any workload with $J$ distinct object sizes, the first-fit allocation algorithm achieves competitive ratio $O(1 + J)$ \cite{randomized_LNO96}.
\end{proof}

The second corollary of Theorem \ref{thm:main-lower} is for the case where $k = 1$. Here, we prove a conjecture of \cite{randomized_LNO96}, showing that randomized allocation algorithms cannot (in an asymptotic sense) beat deterministic ones:

\begin{corollary}\label{cor:main-no-fragments}
   Let $M > 0$.  If $\alg{}$ is a (possibly randomized) algorithm which does not fragment allocations, then there exist workloads $W$ with volume high-water mark $M(W) = M$ such that 
    \[ \comp(\alg,\M) \geq \Omega\left(\log{\M}\right). \]
\end{corollary}

\paragraph{Reducing Theorem \ref{thm:main-lower} to Theorem \ref{thm:main-lower-reduced}.} The first step of the proof, which we defer to \Cref{app:reductions} is the following. We show that to prove \Cref{thm:main-lower}, it is sufficient to prove an alternative theorem that applies to algorithms that use what we call \emph{$\tau$-\trickyalignedevenfrag{$({=}k)$}}.

In what follows, and throughout the rest of the section, we will assume without loss of generality that $k$ satisfies $k \ge 2$ is a power of $2$ and that $M, Q$ are powers of $k$.\footnote{Note that rounding $k$ up to a power of $2$ greater than $1$, and rounding $M$ and $Q$ down to powers of $k$ only makes the lower bound harder to prove, because it does not change the asymptotic lower bound that we wish to achieve.} Moreover, we will only discuss workloads in which the requests are also power-of-two sizes. 

With this in mind, we define a \defn{\alignedevenfrag{$({=}k)$}} algorithm to be one that fragments each request $q$ into \emph{exactly} $k$ fragments, each of which is allocated a sequence of slots of the form $(j \cdot |q| / k, (j + 1) |q| / k]$ for some non-negative integer $j$. 
On top of this, a \defn{$\tau$-\tricky}algorithm is allowed to \emph{ignore} each request (not necessarily independently!) with probability (up to) $\tau$. Requests ignored this way do not count against its memory high-water mark.

\Cref{app:reductions} shows that, to prove Theorem \ref{thm:main-lower}, it suffices to prove the following.

\begin{restatable}{theorem}{mainlowerreduced}\label{thm:main-lower-reduced}
    Let $k \ge 2$ be a power of two, and let $M$ and $Q$ be powers of $k$. If $\alg{}$ is a $1/4$-\trickyalignedevenfrag{$({=}k)$} allocation, then there exists a workload $\workload$ (whose requests all have power-of-two sizes) with $\M(\workload) = O(\advmem)$ and $\Q(\workload) = \advreq$ such that 
    \[ \comp(\alg,\workload) \geq \Omega\left(\log_{k}\advreq\right). \] 
\end{restatable}

The rest of the section is dedicated to proving \Cref{thm:main-lower-reduced}. Throughout, we will assume without loss of generality that the parameters $k, M, Q$ are all powers of $2$.

In the rest of this section, we will prove \Cref{thm:main-lower-reduced}.
We start by defining an adversarial workload and analyzing its properties.
Then we analyze the memory usage of an allocator (satisfying the conditions of \Cref{thm:main-lower-reduced}) against the workload.

\subsection{The Adversarial Workload}

The proof of Theorem \ref{thm:main-lower-reduced} makes use of the following workload $\adv = \adv(\advmem, \advreq)$ with volume high-water mark $\M(\adv) = O(\advmem)$ and request high-water mark $\Q(\adv) = \advreq$. 

The workload $\adv$ proceeds in $R = \log_{k} \advreq + 1$ rounds $r = 1, 2, \ldots, R$, where in each round $r$, all of the requests that are made have the same size 
$$s_r := M k^{r-1}/\advreq.$$ 
Since the allocation algorithm breaks each request into exactly $k$ fragments, it follows that each round-$r$ fragment has size 
$$f_r := s_r/k.$$ 
Note that both $s_r$ and $f_r$ are, by assumption, powers of $k$.
To assist with the discussion of the rounds, we will also use $Q_t(r)$ to denote the set of all round-$t$ requests that remain live after round $r$. 

With this notation in place, the structure of each round $r$ is as follows: 
\begin{itemize}
    \item First, the workload frees a uniformly random $0.1$ fraction of each existing request class $Q_{t}(r-1)$, for $t=1,\ldots,r-1$ (these become $Q_t(r)$).
\item Then, the workload makes $m/s_r$ requests of size $s_r$ as the next request class $Q_r$.
\end{itemize}
Note that $\adv$'s randomization affects only which requests it frees in each round. 

Before continuing into the analysis, we take a moment to remark on some basic properties that the above workload has. First, notice that the volume of requests made in each round is deterministically $M$. As a consequence, we have:

\begin{lemma}\label{lem:total-alloc-size}
The total mass of requests $A$ ever made by $\adv$ is $MR = \Omega(M \log_{k} \advreq)$.
\end{lemma}

Also notice that, by design, we have the desired bounds on $\M(\adv)$ and  $\Q(\adv)$. 
\begin{lemma}\label{lem:alloc-size}
    The volume high-water mark of $\adv$ is $\M(\adv) = O(\advmem)$ and the request high-water mark is $\Q(\adv) = \advreq$.
\end{lemma}

\begin{proof}
    At the end of a given round $r$, the volume of live requests is at most 
    $$\sum_{i = 1}^r M \cdot (0.9^{r - i}) = O(M),$$
    implying a volume high-water mark of $O(M)$. The request high-water mark occurs after the first round, when $\advreq$ requests have been made. 
\end{proof}

Finally, in the proof of \Cref{thm:main-lower-reduced}, we will make use of the following basic bound on the probability that a given set of requests have \emph{all} been freed.

\begin{lemma}\label{lem:na-chernoff}
    Let $q_1, q_2, \ldots, q_\ell$ be requests made in round $r$, and let $X$ be the event that every $q_i$ has been freed by round $t > r$.
    Then \[\Pr[X] \leq 2^{-\Omega(\ell \cdot 0.9^{t-r})}.\]
\end{lemma}
\begin{proof}
   Let $X_t(q)$ be the random variable indicating that request $q$ has been freed by round $t$. Since the round-$r$ requests freed in rounds $r + 1, \ldots, t$ are a random subset of the round-$r$ requests, we know that $X_1(q), X_2(q), \ldots, X_t(q)$ are negatively associated (see, e.g., \cite{joag1983negative,wajc2017negative}), which implies that
   $$\Pr[X_t(q) \forall t \in [\ell]] \le \prod_{t = 1}^\ell \Pr[X_t(q)] = (1 - 0.9^{t - r})^\ell = 2^{-\Omega(\ell \cdot 0.9^{t-r})},$$
   where the final step uses the identity $1 - x = 2^{-\Omega(x)}$. 
\end{proof}

\subsubsection{Proof of Theorem~\ref{thm:main-lower-reduced}}\label{sec:lower-proof-reduced}

Having constructed the workload $\adv$, we turn to proving \Cref{thm:main-lower-reduced}.

Define the \defn{round-$r$ blocks}, denoted $\mathcal{B}_r$, to be the memory-aligned blocks of size $f_r = Mk^{r-2}/\advreq$ in the memory array 
We refer to these as round-$r$ blocks because in round $r$ of the workload, an \alignedevenfrag{$(=k)$} algorithm must allocate fragments precisely into these blocks.

A round-$r$ block consists of $k$ round-$(r-1)$ blocks, $k^2$ round-$(r-2)$ blocks and so on.
Notice that if the algorithm allocates a fragment into a round-$r$ block in round $r$, it entirely consumes the block, so there can have been no live fragments present.
Thus, if there were allocations made to any of its constituent round-$t$ blocks (for $t < r$) in prior rounds, these must have been deallocated already.

A key insight behind the following analysis is that if the component round-$t$ blocks of a round-$r$ block were allocated efficiently, i.e., were packed with fragments, then the probability that all of these round-$t$ fragments have since been deallocated is very small.
Thus the algorithm loses either way: either it uses its space inefficiently up front, or it cannot reuse blocks as deallocations make them sparser.
At a high level, the proof shows the algorithm must allocate a constant fraction of the total allocated volume (excluding deallocations) in new space, which yields the lower bound.

More concretely, for a given round-$r$ block $B$, let $I_B$ be the indicator random variable for whether a $f_r$-size fragment is ever allocated in that block (it counts even if the fragment is later deallocated). Note that if a \tricky algorithm ignores a request, then this request will not contribute to any $I_B$ values. 

Let $c$ be a parameter that we will later set to be a large positive constant. For each round $r$, and for each round-$r$ block $B\in\mathcal{B}_r$ define $\phi(B)$ and $\psi(B)$ as follows.
If $r = 1$, then $\phi(B) = \psi(B) = I_B$.
If $r > 1$, then 
\[\phi(B) = \psi(B_1) + \psi(B_2) + \cdots + \psi(B_{k}) + f_r I_B,\]
where $B_1, \dots, B_{k}$ are the $k$ level-$(r - 1)$ blocks that comprise $B$; and define
\[\psi(B) = \min(\phi(B), c f_r).\]
For each block $B$, define $\delta(B) = \phi(B) - \psi(B)$. The following lemma will be useful later.

\begin{lemma}\label{lem:delta-cap}
    For all rounds $r$ and blocks $B \in \mathcal{B}_r$, $\delta(B) \leq f_r$.
\end{lemma}

\begin{proof}
    Each $\psi(B_i) \leq cs_{r-1}/k$ for the blocks $B_1,\dots B_{k}$ comprising $B$. Thus 
    \begin{align*} \phi(B) \leq k\cdot cs_{r-1}/k + f_r I_B = cf_r + f_r I_B \leq cf_r + f_r\end{align*} 
    and so $\delta = \phi(B) - \min(\phi(B),cf_r) \leq f_r$.
\end{proof}

Before continuing, let us briefly discuss some intuition about how we will proceed.
Intuitively, for a round-$r$ block $B$, we will want a quantity like $\psi(B)$ to be a measure of how efficiently the algorithm has used the space in $B$ in the rounds up to and including $r$.
Define $\psi^*(B) = I_B$ if $r = 1$ and 
\[\psi^*(B) = \psi^*(B_1) + \psi^*(B_2) + \cdots + \psi^*(B_{k}) + f_r I_B\]
if $r > 1$. Equivalently, $\psi^*(B)$ is the total volume of allocations made into $B$ in all rounds up to $r$.
If we could show that $\psi^*(B)$ were  bounded by $O(f_r)$ for all $B$ and $r$ and some constant $c$, this would say that the total volume of allocations into $B$ over time is a constant multiple of its size. But since the total volume of \emph{all requests} in the first $r$ rounds is $\Omega\left(M r\right)$, this would imply that the algorithm must use many round-$r$ blocks, thereby incurring a large memory high-water mark.

Rather than trying to directly prove that $\psi^*(B)$ (or $\E[\psi^*(B)]$) is at most $O(f_r)$, for each round-$r$ block $B$, we will instead focus on the closely related functions $\psi(B)$ and $\phi(B)$ defined above. These functions are defined in such a way that, when $\psi^*(B) \ge c f_r$, we can charge the ``overflow'' to the component sub-blocks of $B$. In particular, when $\phi(B)$ exceeds $cf_r$, we want this to be charged to $B$ and then not recharged to the superblocks of $B$. $\delta(B)$ captures this charge.

The following two lemmas capture this intuition more formally. The first captures the idea that, to lower-bound $\hwm(\alg, \adv)$, it suffices to lower-bound $\sum_{B \in \mathcal{B}_{R}} \psi(B)$.

\begin{lemma} \label{lem:psi-lb-hwm}
    Let $\alg$ be a $1/4$-\trickyalignedevenfrag{$(=k)$} algorithm.
    Then on workload $\adv$,
    \[\sum_{B \in \mathcal{B}_{R}} \psi(B) \le 2c \hwm(\alg,\adv).\]
\end{lemma}

\begin{proof}
    Let $\overline{\mathcal{B}}_{R}$ be the blocks $B \in \mathcal{B}_{R}$ that are used by the algorithm in at least one round.
    By design, $\psi(B) = 0$ for all $B \in \mathcal{B}_{R} \setminus \overline{\mathcal{B}}_{R}$, and $\psi(B) \le c |B|$ for each $B \in \overline{\mathcal{B}}_{R}$. Thus 
    \[ \sum_{B \in \mathcal{B}_{R}} \psi(B) \le \sum_{B \in \overline{\mathcal{B}}_{R}} c|B| \le 2 c \hwm(\alg,\adv). \]
    The factor of 2 in the last inequality comes from the fact that the last block can potentially be almost entirely over the memory high-water mark; the largest block size is $\advmem \le \hwm(\alg, \adv)$, so this consideration contributes at most a factor of $2$. 
\end{proof}

Therefore, if we can lower bound $\sum_{B \in \mathcal{B}_{R}}\psi(B)$, this provides a lower bound on the memory high-water mark $\hwm(\alg,\adv)$ as well. Recall that $1/4$-\tricky algorithms may probabilistically ignore some allocations.
Let $A=A(\alg,\adv)$ be the random variable denoting the total volume of the allocations that are \emph{not} ignored by $\alg$ on $\adv$. (Note that $\E[A] \ge 3/4 \cdot M(V)$, by definition.) The following lemma captures the idea that we can use the $\delta(B)$'s to evaluate the difference between $\sum_{B \in \mathcal{B}_{R}} \psi(B)$ and $A$.

\begin{lemma}\label{lem:psi-delta}
    let $\alg$ be a $1/4$-\trickyalignedevenfrag{$(=k)$} algorithm.
    Then on workload $\adv$,
    \[\sum_{B \in \mathcal{B}_{R}} \psi(B) = A - \sum_{r = 1}^{R} \sum_{B \in \mathcal{B}_r} \delta(B).\]
\end{lemma}

\begin{proof}
    Let $A_r$ be the total volume of unignored requests in round $r$.
    We proceed by induction on $r$.
    For round $1$, we have
    \[ \sum_{B \in \mathcal{B}_{0}} \psi(B) = \sum_{B \in \mathcal{B}_{1}} I_B = A_1. \]
    Assume that
    \[\sum_{B \in \mathcal{B}_{r-1}} \psi(B) = \sum_{i=1}^{r-1}A_{i} - \sum_{i = 1}^{r-1} \sum_{B \in \mathcal{B}_i} \delta(B).\]
    For a round-$r$ block $B$, $\phi(B)$ is the sum of $\psi$ on its round-$(r-1)$ subblocks together with the volume any new allocation made to $B$ during the round. In round $r$, $A_r$ total allocations are made, so 
    \begin{align*}
        \sum_{B \in \mathcal{B}_{r}} \psi(B) &= \sum_{B \in \mathcal{B}_{r}} \phi(B) - \sum_{B \in \mathcal{B}_{r}} \delta(B) \\
        &= A_r + \sum_{B \in \mathcal{B}_{r - 1}} \psi(B)  - \sum_{B \in \mathcal{B}_{r}} \delta(B) \\
        &= A_r + \sum_{i=1}^{r-1}A_{i} - \sum_{i = 1}^{r-1} \sum_{B \in \mathcal{B}_i} \delta(B) - \sum_{B \in \mathcal{B}_{r}} \delta(B) \\
        &= \sum_{i=1}^{r}A_{i} - \sum_{i = 1}^{r} \sum_{B \in \mathcal{B}_i} \delta(B),
    \end{align*}
    which completes the induction. The result now follows since $A = \sum_{i=1}^R A_i$.
\end{proof}

Taking expected values, we can deduce:
\begin{corollary}\label{cor:psi-delta-tricky}
    Let $\alg$ be a $1/4$-\trickyalignedevenfrag{$(=k)$} algorithm.
    Then on workload $\adv$,
    \[ \E\left[\sum_{B \in \mathcal{B}_{R}} \psi(B)\right] \geq \frac{3}{4} RM - \E\left[\sum_{r = 1}^{R} \sum_{B \in \mathcal{B}_r} \delta(B)\right].\]
\end{corollary}

\begin{proof}
    Since $\alg$ is $1/4$-\tricky, for any workload $W$ with request volume $x$, the allocated volume $A(\alg,W)$ must satisfy $E[A(\alg,W)] \geq \frac{3}{4}x$.
    The lemma follows by taking expectations in \Cref{lem:psi-delta}, and \Cref{lem:total-alloc-size}.
\end{proof}

We remark that the $1/4$-\trickiness of $\alg$ comes into play in the above lemma, when taking $\E[A]$. It turns out that this will be the only point in the proof where the $1/4$-\trickiness ends up mattering in any substantive way. 

At this point, we have established a chain of connections, reducing the task of proving \Cref{thm:main-lower-reduced} to the task of bounding the sum of the $\delta$s to be at most, say, $RM / 4$. This is what we do in the following lemma, which is the main technical lemma of the section:

\begin{lemma} \label{lem:delta-small}
    If we set $c$ to be a sufficiently large positive constant, then
    \[\E\left[ \sum_{r = 1}^{R} \sum_{B \in \mathcal{B}_r} \delta(B) \right] \le \frac{RM}{4}.\]
\end{lemma}

\begin{proof}
    Consider a round-$t$ block $B$ with $\delta(B) > 0$.
    Then by definition, $\psi(B) > cf_t$, and in particular, the sum of all allocations made within $B$ (including those since deallocated) is at least $cf_t$.
    Since at most $f_t$ volume can be allocated in $B$ in round $t$, this means that at least $cf_t - f_t = (c-1)f_t$ volume must have been allocated in $B$ in prior rounds.
    Let $d = (c-1)/\left(\sum_{r=1}^{t-1} \frac{1}{1.1^{t-r}}\right)$.
    Then
    \[ (c-1)f_t = d\sum_{r=1}^{t-1}\frac{f_t}{1.1^{t-r}} = \sum_{r=1}^{t-1}\frac{dk^{t-r}f_r}{1.1^{t-r}},\]
    so there must be some round $r < t$ such that the number of fragments allocated in $B$ in round $r$ is at least $\frac{dk^{t-r}}{1.1^{t-r}}$.
    Let $r(B)$ be the largest such $r$ (i.e. the most recent round in which this occurred), and ``charge'' $\delta(B)$ to the pair $(r(B), t - r(B))$. 
    
    For each pair $(r, i)$ define $C_{r, i}$ to be the total amount charged to that pair by all round-$(r + i)$ blocks. By construction,
    \[ \sum_{r = 0}^{R} \sum_{B \in \mathcal{B}_r} \delta(B) = \sum_{(r, i)} C_{r, i}, \]
    so we will focus the rest of the proof on bounding the expectation of the latter sum.

    Let $\mathcal{B}_{r, i}$ be the set of blocks $B \in \mathcal{B}_{r + i}$ with the property that $B$ received at least $d k^i /  1.1^i$ fragments from round-$r$ requests. 
    Note that since $\alg$ allocates at most $M/f_r$ fragments in round $r$,
    \begin{equation} \label{eq:blocks}
    |\mathcal{B}_{r, i}| \le  \frac{M/ f_r}{d k^i / 1.1^i}  = \frac{M1.1^i}{dk^if_r} = \frac{M1.1^i}{df_{r+i}}.
    \end{equation}
    Any block that charges value to $C_{r , i}$ must be in $\mathcal{B}_{r, i}$. Furthermore, in order for a block $B \in \mathcal{B}_{r, i}$ to charge value to $C_{r, i}$, all of the $d k^i / 1.1^i$ size-$f_r$ fragments in $B$ must get deallocated by the end of round $r + i$. These fragments come from at least $d k^{i-1} /  1.1^i$ \emph{unique requests}, since each request gets exactly $k$ fragments.
    Let $Y_{B, r, i}$ be the number of \emph{live} round-$r$ requests in $B$ at the end of round $r+i$.
    
    By \Cref{lem:na-chernoff} the probability that all these requests are freed by round $r + i$ is
    
    \begin{equation} \label{eq:prob}
    \Pr[Y_{B, r, i} = 0] \leq 2^{-\Omega(d k^{i-1} 0.9^i / 1.1^i)}.
    \end{equation}
    
    Combining \eqref{eq:blocks} and \eqref{eq:prob}, the expected number of blocks in $\mathcal{B}_{r, i}$ that charge anything to $C_{r, i}$ is at most 
    \begin{align*}
        & |\mathcal{B}_{r, i}| \cdot 2^{-\Omega(d k^{i-1} 0.9^i / 1.1^i)} \\
        & \leq  \frac{M1.1^i}{df_{r+i}} 2^{-\Omega(d k^{i-1} 0.9^i /  1.1^i)}. \\
        \end{align*}
    If a block $B$ does charge $\delta(B)$ to $C_{r, i}$, then since $\delta(B) \le f_{r+i}$ by Lemma~\ref{lem:delta-cap}, the charge can be at most $f_{r+i}$.
    
    Thus,
    \[\E[C_{r, i}] \le \frac{m 1.1^i}{d} 2^{-\Omega(d k^{i-1} 0.9^i / 1.1^i)}.\]
    Because $k \geq 2$, $k^{i-1} 0.9^i / 1.1^i \geq i$ for large enough $i$.
    Therefore, for $c$ large enough (and therefore $d$ large enough), this implies
    \[\E[C_{r, i}] \le \frac{m}{8 \cdot 2^i}.\]
    Finally, summing over all pairs $(r, i)$, we get
    \begin{align*} 
            \E\left[\sum_{(r, i)} C_{r, i}\right] & \le \sum_{r \in [0, R]} \sum_{i \ge 0} \frac{\m}{8 \cdot 2^i} \\
        & \le \sum_{r \in [0, R]} M/ 4 \\
        & \le m R / 4.
    \end{align*}
\end{proof}

\begin{proof}[Proof of \Cref{thm:main-lower-reduced}]
By \Cref{lem:psi-lb-hwm}, 
\[\hwm(\alg,\adv) \geq \Omega\left(\sum_{B \in \mathcal{B}_{R}} \psi(B)\right),\]
which by \Cref{cor:psi-delta-tricky} is at least
\[\frac{3}{4}MR - \sum_{r = 1}^{R} \sum_{B \in \mathcal{B}_r} \delta(B).\]
By \Cref{lem:delta-small,lem:alloc-size}, this has expected value at least
\[\frac{3}{4}MR - \frac{1}{4}MR \ge \Omega(MR) = \Omega(\advmem \log_k(\advreq)).\]
This completes the proof of the theorem.
\end{proof}
\section{Lower Bounds for Aggregate Fragmentation}\label{sec:agg}

In this section, we develop (what will turn out to be) tight lower bounds for algorithms using $k$-aggregate fragmentation. As we will see, these lower bounds depend heavily on the machinery already developed in Section \ref{sec:lower-per-request}. 

We begin with the following bound, which is tight for $k \ge 1 + \Omega(1)$. 
\begin{restatable}{theorem}{fragmento}\label{thm:fragmento}
    For any $k$-expected-aggregate fragmentation algorithm $\alg$ with $k > 1$, and $\Mmin,\Mmax$ such that $\Mmin \leq \Mmax$, 
    there exists $\workload$ with $\M(W) \in [\Mmin, \Mmax]$ for which 
    \[\comp(\alg,\workload) \geq \Omega(\log_{\max(k, 2)}\log(M_{max}/M_{min})).\]
\end{restatable}
\begin{proof}
We may assume WLOG that $k \ge 2$. We may also assume WLOG that $\log(\Mmax/\Mmin) \leq \Mmin$ since increasing $\Mmin$ so that $\log(\Mmax/\Mmin) \leq \Mmin$ does not asymptotically change the quantity $\log (\Mmax / \Mmin)$, and therefore does not asymptotically change the guarantees of the theorem. 

Let $\Lambda = \Mmax / \Mmin$. Let $\alg$ be an algorithm with $\E[\hwm(\alg,\workload)] \leq o(\Mwork \log_k\log \Lambda)$ for all workloads $\workload$ with $\M(\workload) \in [\Mmin, \Mmax]$.
We will show that there exists a workload $\workload$, with $\M(\workload) \in [\Mmin, \Mmax]$, with some request high-water mark $\Q$ at the end of which $\alg$ has more than $kQ$ (expected) fragments that are simultaneously live, therefore contradicting the fact that $\alg$ uses $k$-expected aggregate fragmentation. Throughout the proof, we will use $\fragnum(r, \mathcal{A})$ for a request $r$ to denote the number of fragments that the algorithm $\mathcal{A}$ breaks $r$ into.

    We iteratively use \Cref{thm:main-lower} in rounds to build the desired workload $W$.
    Set $R = \log{\Lambda}$ and $M_r = \Mmin\cdot 2^{r-1}$ for $r=1,\ldots,R$.
    By \Cref{thm:main-lower}, there exists a workload $V_r$ with $\M(V_r) = M_r$ and $\Q(V_r) = R$ such that any algorithm with $\E[\hwm(\alg,V_r)] \leq o(\M(V_r) \log_{k+1} \Q(V_r)) = o(M_r \log_{2k+1} R)$ must have some request $q_r$ with $\E[\fragnum(q_r,\alg)] > 2k$. In round $R$ of the workload $W$, we run $V_r$ until we perform a request $q_r$ satisfying $\E[\fragnum(q_r,\alg)] > 2k$ (we will argue shortly that such a request must occur). When $q_r$ is requested, we stop running $V_r$ and we free all other requests from $V_r$, ending the round. 
    
    To analyze the workload $W$, we must first argue that each $q_r$ exists. To see this, observe that the volume high-water mark from rounds $1$ to $r$ is at most $\sum_{i=1}^r M_i \leq 2 M_r$.
    Therefore, by assumption, $\alg$ has memory high-water mark $o(\Mwork \log_k\log \Lambda)  \leq o(2M_r \log_{2k+1}  R)$. It follows by the definition of $V_r$ that there does, in fact, exist some request $q_r$ in round $r$ satisfying $\E[\fragnum(q_r,\alg)] > 2k$.

    Having proven that the requests $q_1, q_2, \ldots, q_R$ exist, we can now obtain a contradiction as follows. At the end of round $R$, the live requests are $q_1, q_2, \ldots, q_R$ (we call these the \emph{retained requests}). Recalling that each $V_r$ has request high-water mark $Q(V_r) = R$ (by construction), the overall request high-water mark $Q(W)$ is at most $2R$ (accounting both for the high-water mark of a given $V_r$, plus the up to $R$ retained requests that may also be live). Each retained request $q_i$ has $\E[\fragnum(q_i, \mathcal{A})] > 2k$, so the fragment high-water mark at the end of round $R$ is more than $2Rk$ in expectation. This contradicts the fact that $\alg$ uses expected-$k$ aggregate fragmentation.

\end{proof}

In the case where $k = 1 + \epsilon$ for some $\epsilon \in (0, 1)$, we can prove a separate lower bound of $\Omega(\log \epsilon^{-1})$ on the competitive ratio. 
\begin{lemma}
Consider parameters $\epsilon \in (0, 1)$ and $M \in \mathbb{N}$ satisfying $M \ge \epsilon^{-1}$. Let $\mathcal{W}$ denote the set of workloads $W$ with volume high-water mark $M(W) = \Theta(M)$. Any (randomized) online algorithm using $(1 + \epsilon)$-expected aggregate fragmentation must incur a competitive ratio of $\Omega(1 + \log \epsilon^{-1})$ on some workload $W \in \mathcal{W}$. 
\label{lem:logepslower}
\end{lemma}
\begin{proof}
Note that the algorithm can, in expectation, create at most $O(1)$ partial fragmentations on any workload $W$ with $Q(W) < \epsilon^{-1}$. On such a workload, the algorithm is an $O(1)$-expected per-request fragmentation algorithm. The claimed result therefore follows from Theorem \ref{thm:main-lower}.
\end{proof}

We can now prove \Cref{cor:epslower-final}. Together, \Cref{thm:fragmento} and \Cref{cor:epslower-final} imply that all of our upper bounds in Section \ref{sec:aggupper} are tight.

\epslower*

\begin{proof}
The corollary follows by combining \Cref{lem:logepslower} and \Cref{thm:fragmento}, noting that \Cref{thm:fragmento} immediately applies to the regime where $1 \leq k \leq 2$ (since this gives a stronger condition on the algorithm than $k\geq 2$). 
\end{proof}
\section{Upper Bounds for Aggregate Fragmentation}\label{sec:aggupper}

In this section, we construct optimal allocation algorithms using $k$-aggregate fragmented allocation. All of the bounds in this section are asymptotically tight by \Cref{thm:fragmento} and \Cref{cor:epslower-final}. 

We begin with the setting where $M(W)$ is known, and where $k = 1 + \epsilon$ for some $\epsilon \in (0, 1)$. The algorithm that we construct in this setting (which comes with a surprisingly subtle analysis) ends up serving as the main building block for all of the other algorithms that we construct later on. 

\begin{theorem}[$(1 + \epsilon)$-Aggregate-Fragmented Allocation with Known $M$]
Consider parameters $\epsilon \in (0, 1)$, $M \in \mathbb{N}$, and $q \ge 2$. Let $\mathcal{W}$ denote the set of workloads $W$ with volume high-water mark $M(W) = \Theta(M)$, and where every allocation has a power-of-$q$ size. There is a (deterministic) online allocation algorithm $\mathcal{A}^\epsilon_M$, using $(1 + \epsilon)$-aggregate fragmentation, that achieves competitive ratio $O(1 + \log_q \epsilon^{-1})$ on every workload $W \in \mathcal{W}$. 
\label{thm:knownMupper}
\end{theorem}

\begin{proof}
The algorithm $\mathcal{A}^\epsilon_M$ keeps track, at all times $t$, of the maximum number $\Qguess(t)$ of requests that are ever live simultaneously at or prior to time $t$. If, at time $t$, a request of size $s$ is made, then the algorithm considers two cases:
\begin{enumerate}
    \item If $s \le \epsilon^{-1} M / \Qguess(t)$, then the request is kept as a single fragment. 
    \item If $s > \epsilon^{-1} M / \Qguess(t)$, then the request is broken into $\lceil \epsilon s \cdot \Qguess(t) / M\rceil$ fragments of size at most $\epsilon^{-1} M / \Qguess(t)$. 
\end{enumerate}
Each fragment is then allocated according to the first-fit algorithm, which places the fragment in the earliest free region in memory whose size is large enough for the fragment. 

To simplify exposition for how we discuss fragments, define the \defn{effective size} $s(f)$ of each fragment $f$ created at time $t$ to be $|f|$ if the fragment is created by Case 1, and to be $\epsilon^{-1} M / \Qguess(t)$ if the fragment is created by Case 2. Note that the sum of the effective sizes of all live fragments is at most $2M$. 

Before analyzing the memory high-water mark, let us take a moment to verify that the algorithm $\mathcal{A}(M)$ uses at most $(1 + \epsilon)$ aggregate-fragmentation. Let $\Q$ be the largest number of requests that are ever simultaneously live. By design, the algorithm breaks each request of a given size $s$ into at most $$1 + \epsilon s \cdot \Q / M$$
fragments. Summing over the set $L$ of live requests at any given moment, the total number of fragments is at most
$$|L| + \sum_{r \in L} \epsilon |r| \cdot \Q / M \le |L| + \epsilon Q.$$
Since there are $|L|$ non-partial fragments, the number of partial fragments is at most $\epsilon Q$, as desired. 

In the rest of the proof, we bound the competitive ratio achieved by the algorithm. For $j \ge 0$ and time $t$, define $\ell_j(t)$ to be the largest power of $q$ smaller than $q^j \cdot M / \Qguess(t)$. Say that a fragment $f$ is \defn{type-$j$} if, at the time $t$ when $f$ was allocated, $j$ was the smallest $j \ge 0$ such $|f| \le \ell_j(t).$
Notably, if $|f|$ has size at most $\ell_0(t)$, then its type will be $0$. Additionally, break memory into \defn{segments}, where each segment has size $4M$.

To analyze this algorithm, we establish the following \defn{type-segmentation invariant}: Every type-$j$ fragment is placed somewhere in the first $j + 1$ segments of memory. Note that this invariant is enough to complete the proof---since every fragment has type at most $\log_q \epsilon^{-1} + 1$, it follows that the memory-high water mark $S(\mathcal{A}(M), W)$ is at most $O(M (1 + \log_q \epsilon^{-1}))$.

We emphasize that, in the above invariant, the type $j$ associated with a fragment $f$ is partially a function of \emph{when} $f$ was inserted. (If the same fragment $f$ were inserted earlier/later, it might have a completely different type.) Given this, it may seem odd that one can state an invariant that treats all objects of a given type in the same way -- nonetheless, we will see that the invariant holds. 

We prove the type-segmentation invariant by strong induction on $t$. Suppose that the invariant holds for all prior allocations, and consider a fragment $f$ allocated at time $t$.

The first case to consider is if $f$ has type $0$. Suppose for contradiction that $f$ does not get allocated in the first segment. Let $P$ denote the set of fragments that \emph{are} allocated in the first segment. Since $f$ did not get allocated in the first segment, the segment does not contain any free interval of size $M / \Qguess(t)$. But this means that the total size of the first segment is less than
$$M/\Qguess(t) + \sum_{p \in P} (|p| + M/ \Qguess(t)).$$% \le M + (|P| + 1) \cdot M / \Qguess(t).$$

To evaluate this summation, we will calculate an upper bound on the number $|P|$ of fragments in the first segment. By the definition of the algorithm and since $\Qguess(t)$ can only increase over time, the number of fragments of size less than $\epsilon^{-1} M / \Qguess(t)$ is at most $\Qguess(t)$. Thus, the number of fragments in the first segment is $|P|\leq \Qguess(t) + M/(\epsilon^{-1} M / \Qguess(t)) = \Qguess(t)\cdot (1+\epsilon)$. This means that the size of the first segment is less than $M + (|P| + 1) \cdot M / \Qguess(t) < 4M$.  But we also know that the first segment has size exactly $4M$, so we get a contradiction. 

The second case to consider is if $f$ has type $i$ for some $i > 0$. In this case, we claim that all of the fragments $f'$ currently in segment $i + 1$ have effective sizes $s(f')$ satisfying $|f| \le s(f')$. By the invariant, such a fragment $f'$ must have some type $j \ge i$. Let $t' \le t$ be the time at which $f'$ was inserted. Here there are two cases: if $|f'|$ is a power of $q$, then it must be that $|f'| = \ell_j(t')$, meaning that $s(f') \ge |f'| \ge \ell_j(t') \ge \ell_j(t)$; in this case, since $|f| \le \ell_j(t)$, we get that $|f| \le s(f')$, as desired. The other case is that $s(f') = \epsilon^{-1} \cdot M / Q(t') \ge \epsilon^{-1} \cdot M / \Qguess(t)$; but, since $f$ is allocated at time $t$, its size $|f|$ is at most $\epsilon^{-1} \cdot M / \Qguess(t)$, so once again we have $|f| \le s(f')$. Thus every $f'$ currently in segment $i + 1$ has effective size $s(f') \ge |f|$. 

Now, suppose for contradiction that $f$ does not get placed in the first $i + 1$ segments. Then, in segment $i + 1$, every free interval has size at most $|f|$. If we define $P$ to be the set of fragments in the segment, then it follows that the segment has size at most
$$|f| + \sum_{p \in P} (|p| + |f|).$$
Since $|f| \le s(p)$ for each $p \in P$, this is less than
$$|f| + \sum_{p \in P} (|p| + s(p)).$$
Since $|f| \le M$, $\sum_{p \in P} |p| \le M$, and $\sum_{p \in P} s(p) \le 2M$, it follows that that the size of the $(i + 1)$-th segment is less than $4M$, which once again is a contradiction.
\end{proof}

Theorem \ref{thm:knownMupper} assumes that every request has a power-of-$q$ size. Notice, however, that for $q = 2$, this assumption is WLOG (up to a factor-of-two change to $M$). Thus we get the following corollary, which does not make any assumptions on the sizes of requests.
\begin{corollary}
Consider parameters $\epsilon \in (0, 1)$ and $M \in \mathbb{N}$. Let $\mathcal{W}$ denote the set of workloads $W$ with volume high-water mark $M(W) = \Theta(M)$. There is a (deterministic) online allocation algorithm $\mathcal{A}^\epsilon_M$, using $(1 + \epsilon)$-aggregate fragmentation, that achieves competitive ratio $O(1 + \log \epsilon^{-1})$ on every workload $W \in \mathcal{W}$. 
\label{cor:Mknownupper}
\end{corollary}

Next, we turn our attention to the setting where $M$ is not known (or is only partially known) a priori. Here, we consider two parameter regimes for $k$, the \emph{large regime} of $k \ge 2$ (Theorem \ref{thm:unknownMupper}) and the \emph{small regime} of $k \in (1, 2)$ (Theorem \ref{thm:unknownMuppereps}). As noted earlier, in each case, the competitive ratios that we obtain are optimal (by Theorem \ref{thm:fragmento} and \Cref{cor:epslower-final}). 

\begin{theorem}[$k$-Aggregate Fragmented Allocation with Unknown $M$]
Consider parameters $k \ge 2$ and $M_0 \ge 1$. There is a (deterministic) online allocation algorithm using $k$-aggregate fragmentation that achieves competitive ratio $O(1 + \log_k \log (M(W) / M_0(W)))$ on every $W$ satisfying $M(W) \ge M_0$. 
\label{thm:unknownMupper}
\end{theorem}

\begin{proof}
Note that we can use $0.1 k$-aggregate fragmentation to break requests into fragments whose sizes are each within a constant factor of a power of $k$. Thus, rather than using $k$-aggregate fragmentation with arbitrarily-sized requests, we may assume without loss of generality that we are using $0.9k$-aggregate fragmentation and that every request size is a power of $k$. We will include this assumption throughout the rest of the proof.\footnote{In fact, obtaining this assumption is the only part of the proof in which we require $k$-aggregate fragmentation as opposed to, say, $1.8$-aggregate fragmentation.} We will also assume without loss of generality that $M_0$ is a power of $2$.

With these assumptions in place, we define the algorithm $\mathcal{A}$ as follows. At all times $t$, keep track of a quantity $\Mguess(t)$, defined to be $M_0$ if the volume high-water mark up to time $t$ is at most $M_0$, and defined otherwise to be the volume high-water mark up to time $t$. Whenever $\Mguess(t)$ reaches (or crosses) a power of two $M_0 \cdot 2^{j - 1}$, for some $j \ge 1$, we enter what we call \defn{phase $j$}. During phase $j$, we process requests using the algorithm $B_j := \mathcal{A}^{1/j^2}_{M_0 \cdot 2^{j}}$ (as defined in Theorem \ref{thm:knownMupper}). If, prior to phase $j$, the memory high-water mark for $\mathcal{A}$ was given by some quantity $S_j$, then during phase $j$ the algorithm $B_j$ uses the region of memory starting immediately after $S_j$. 

By Theorem \ref{thm:knownMupper}, with $q = k$, the growth $S_{j + 1} - S_j$ in memory high-water mark during phase $j$ is at most
$$O(M_0 \cdot 2^j (1 + \log_k j^2)) = O(M_0  \cdot 2^j (1 + \log_k j)).$$
Summing over $j \in [1, \log (2M(W) / M_0)]$ gives a total memory high-water mark of 
$$\sum_{j \in [1, \log (2M(W) / M_0)]} O(M_0 \cdot 2^j (1 + \log_k j)) = O(M(W) (1 + \log_k \log (M(W) / M_0))),$$
as desired.

Now, let us turn to bounding the total number of partial fragments. Let $Q_j$ (resp.~$Q$) be the maximum number of simultaneously live requests that are allocated by $B_j$ (resp.~by the entire algorithm). Then, by Theorem \ref{thm:knownMupper}, the total number of partial fragments that are ever simultaneously live in $B_j$ is at most 
$$Q_j / j^2.$$
The total number of partial fragments that are simultaneously live across the entire algorithm is therefore at most
$$\sum_j Q_j / j^2 \le Q \sum_j 1/j^2 \le 1.7 Q \le 0.9 k Q,$$
as desired.
\end{proof}

\begin{theorem}[$(1 + \epsilon)$-Aggregate Fragmented Allocation with Unknown $M$]
Consider parameters $\epsilon \in (0, 1)$ and $M_0 \ge 1$. There is a (deterministic) online allocation algorithm $\mathcal{A}$, using $(1 + \epsilon)$-aggregate fragmentation that achieves competitive ratio $O(1 + \log \log (M(W) / M_0(W)) + \log \epsilon^{-1})$ on every workload $W$ satisfying $M(W) \ge M_0$. 
\label{thm:unknownMuppereps}
\end{theorem}
\begin{proof}
We assume without loss of generality that every request has a power-of-two size. We then define the algorithm $\mathcal{A}$ as follows. At all times $t$, keep track of a quantity $\Mguess(t)$, defined to be $M_0$ if the volume high-water mark up to time $t$ is at most $M_0$, and defined otherwise to be the volume high-water mark up to time $t$.  Whenever $\Mguess(t)$ reaches (or crosses) a power of two $M_0 \cdot 2^{j - 1}$, for some $j \ge 1$, we enter what we call \defn{phase $j$}. During phase $j$, we process requests using the algorithm $B_j := \mathcal{A}^{\epsilon/(2j^2)}_{M_0 \cdot 2^{j}}$ (as defined in Theorem \ref{thm:knownMupper}). If, prior to phase $j$, the memory high-water mark for $\mathcal{A}$ was given by some quantity $S_j$, then during phase $j$ the algorithm $B_j$ uses the region of memory starting immediately after $S_j$. 

By Theorem \ref{thm:knownMupper}, with $q = 2$, the growth $S_{j + 1} - S_j$ in memory high-water mark during phase $j$ is at most
$$O(M_0 \cdot 2^j \cdot (1 + \log (2 j^2 / \epsilon))) = O(M_0 \cdot 2^j (1 + \log j + \log \epsilon^{-1}).$$
Summing over $j \in [1, \log (2M(W) / M_0)]$ gives a total memory high-water mark of 
$$\sum_{j \in [1, \log (2M(W) / M_0)]} O(M_0 \cdot 2^j (1 + \log j)) = O(M(W) (1 + \log \log (M(W) / M_0) + \log \epsilon^{-1})),$$
as desired.

Finally, we bound the total number of partial fragments. Let $Q_j$ (resp.~$Q$) be the maximum number of simultaneously live requests that are allocated by $B_j$ (resp.~by the entire algorithm). Then, by Theorem \ref{thm:knownMupper}, the total number of partial fragment that are ever simultaneously live in $B_j$ is at most 
$$Q_j\cdot \epsilon / (2j^2).$$
Thus the total number of partial fragments that are simultaneously live across the entire algorithm is at most
$$\sum_j Q_j \cdot \epsilon / (2 j^2) \le \frac{Q \epsilon}{2} \sum_j 1/j^2 < \epsilon Q,$$
as desired.
\end{proof}

\section*{Acknowledgements}
This work was supported by NSF grants
CNS-2504470,
CCF-2106999,
CCF-2118620,  
CCF-2118832, 
CCF-2106827, and 
CCF-2247577. 
Any opinions, findings, and conclusions or recommendations expressed in this material are those of the authors and do not necessarily reflect the views of the National Science Foundation.

Michael Bender was supported in part by the John L. Hennessy Chaired Professorship and Sandia National Laboratories.
Mart\'{\i}n Farach-Colton was supported in part by the Leonard J. Shustek professorship.
Hanna Koml\'os was supported in part by the Graduate Fellowships for STEM Diversity. This work was done in part while Hanna Koml\'os was visiting the Simons Institute for the Theory of Computing, and was supported in part by Google Research, Fall 2025.
William Kuszmaul was supported in part by a grant from Jane Street. Nicole Wein is supported by NSF CAREER award 2541910. This work was initiated while Nicole Wein was at DIMACS, Rutgers University, supported by a grant to DIMACS from the Simons Foundation (820931), and continued while visiting the Simons Institute.

\bibliographystyle{plain}
\bibliography{references}

\begin{appendix}
\section{Generalizing to Average Fragmentation}
\label{app:average}
In this section, we consider two natural notions of \emph{average fragmentation} as relaxations of per-request fragmentation. We provide reductions showing that they have equivalent bounds to per-request fragmentation. The idea of the reductions is to define workloads that ``artificially'' inflate these measurements of average fragmentation.

Define \defn{memory-average fragmentation} as the high-water mark of the number of live fragments divided by the number of live requests.  Our goal is to prove the following lemma:

\begin{claim} Any algorithm $\alg$ with expected memory-average fragmentation $\leq k$ also must have expected per-request fragmentation $\leq k$. 
\end{claim}

\begin{proof} Suppose for contradiction that there is a workload $\workload$ such that under algorithm $\alg$ there exists a request $q$ with expected $>k$ fragments. We will show that there is another workload $\workload'$ under which $\alg$ incurs expected memory-average fragmentation $>k$. Define $\workload'$ to be the same as $\workload$ until right after $q$. (Note that the oblivious adversary constructing $\workload'$ knows the identity of $q$ by knowing the distribution of $\alg$.) 
Then, $\workload'$ frees all other requests except $q$. Now, the expected memory-average fragmentation is $>k$, a contradiction.
\end{proof}

Define \defn{cumulative-average fragmentation} as the total number of fragments ever created divided by the total number of requests ever allocated. Our goal is to prove the following lemma:
\begin{claim}
    If there is a workload $\workload_{per}$ on which any algorithm with expected per-request fragmentation $\leq k$ has expected memory high-water mark $\hwm$, then there is a workload $\workload_{av}$ for which any algorithm with expected cumulative-average fragmentation $\leq k$ also has expected memory high-water mark $\hwm$.
\end{claim} 

\begin{proof} Let $\alg_{av}$ be an algorithm with cumulative-average fragmentation $\leq k$. In particular, $\alg_{av}$ has the property that for any workload and any request $q$, if the workload were to stop at any point and repeatedly request and frees $q$, then after a finite number of requests/frees, $q$ would have expected number of fragments $\leq k$. (Otherwise, such a workload with an indefinite number of repeated requests/frees of $q$ would have expected cumulative-average fragmentation $>k$.) 

Given the algorithm $\alg_{av}$, let workload $\workload_{av}$ be identical to $\workload_{per}$, except whenever $\alg_{av}$ incurs $>k$ expected fragmentation of a request $q$, $\workload_{av}$ repeatedly requests/frees $q$ until the expected fragmentation of $q$ is $\leq k$, and then continues with workload $\workload_{per}$ as usual. (Note that the oblivious adversary constructing $\workload_{av}$ knows the identity of $q$ by knowing the distribution of $\alg_{av}$.)

Let $\alg_{per}$ be an algorithm identical to $\alg_{av}$, except if $\alg_{per}$ is ever about to service a request $q$ with $>k$ fragments, instead it pretends that the workload repeatedly requests and frees $q$, simulating $\alg_{av}$ until $q$ has $\leq k$ fragments, and requests $q$ with its final division into $\leq k$ fragments. We know that $\alg_{av}$ will eventually break $q$ into $\leq k$ fragments by the previous paragraph. Observe that $\alg_{per}$ has per-request fragmentation $\leq k$.

By construction,  $\alg_{per}$ on $\workload_{per}$, and  $\alg_{av}$ on $\workload_{av}$ have the same expected memory high-water mark. By assumption, running $\alg_{per}$ on $\workload_{per}$ has expected memory high-water mark $S$, so running $\alg_{av}$ on $\workload_{av}$ also has expected memory high-water mark $S$. The claim follows since $\alg_{av}$ was chosen to be an arbitrary algorithm with cumulative-average fragmentation $\leq k$.
\end{proof}
\section{Reducing \Cref{thm:main-lower} to \Cref{thm:main-lower-reduced}}\label{app:reductions}

In this section, we prove that \Cref{thm:main-lower} follows from \Cref{thm:main-lower-reduced}.

In what follows we will always assume (WLOG) that $k \ge 2$, that $k$ is a power of two, and that $M$ and $Q$ are powers of $k$. These assumptions are WLOG in that, even if we round $k$ up to a power of $2$ (greater than $1$) and round $\advmem$ and $\advreq$ down to powers of $k$, the asymptotic lower bound $\Omega \left(\log_{\max\{k,2\}}\advreq\right)$ that we wish to prove does not change.

We now prove a series of reductions between different classes of algorithms.  In what follows, recall that an algorithm is $\tau$-selective if it is permitted to ignore each request with probability up to $\tau$. 

\begin{lemma}\label{lem:tricky}
    Let $\alg$ be an algorithm that breaks each request into $\le k$ fragments, in expectation.
    Then there exists a $1/4$-\tricky algorithm $\balg$ that breaks each request into at most $4 k$ fragments (deterministically!) and such that $\hwm(\balg,\workload) \leq \hwm(\alg, \workload)$ for all workloads $\workload$.
\end{lemma}

\begin{proof}
    We simulate $\alg$ and use its output to create $\balg$.
    Whenever $\alg$ allocates a request with at most $4k$ fragments, $\balg$ performs the same allocation.
    Otherwise, $\balg$ ignores the request, which by Markov's Inequality can happen at most $1/4$ of the time.
\end{proof}

\begin{lemma}\label{lem:equal-size}
    Let $\alg$ be a $1/4$-\tricky algorithm that breaks each request into at most $k$ pieces. 
    Then there exists a $1/4$-\tricky algorithm $\balg$ that breaks each request into $k$ \emph{equal-sized} fragments, and such that for all workloads $\workload$, $\hwm(\balg,\workload) \leq \hwm(\alg, 4\workload)$, where $4\workload$ refers to the workload which allocates requests of exactly $4$ times the size of those in $\workload$.
\end{lemma}

\begin{proof}
    We will simulate $\alg$ and use its output to create $\balg$.
    
    To allocate a request $q$ of size $\size{q}$ in $\balg$, we simulate an request $q'$ of size $4\size{q}$ in $\alg$.
    If $\alg$ ignores the requests, then $\balg$ does too.
    Otherwise, suppose that $\alg$ breaks $q'$ into $\ell$ fragments.
    In $\alg$'s allocation, we will call fragments of size at least $2\size{q}/\ell$ \defn{usable}.
    Then at least $4\size{q}/2=2\size{q}$ of the total mass must be allocated in usable fragments, since at most $\ell\cdot2\size{q}/\ell=2\size{q}$ mass can be allocated in unusable fragments.
    Now, a usable fragment of size $F$ can be divided into $\lfloor F\ell/\size{q} \rfloor$ fragments of size $\size{q}/\ell$, with at most $F/2$ mass remaining.
    Therefore, we can construct at least $\ell$ fragments of size $\size{q}/\ell$ within $\alg$'s allocation.
    We define $\balg$ to allocate $q$ into the first $\ell$ of these memory locations.
    
    $\balg$ ignores a request with the same probability as $\alg$ and breaks each request into equally sized fragments by construction.
    Because each allocation of $\balg$ during $\workload$ is a subset of the allocations made in $\alg$ on $4\workload$, $\hwm(\balg,\workload) \leq \hwm(\alg, 4\workload)$.
\end{proof}

To give the class of algorithms that we have reduced to so far a concrete name, say that an algorithm uses \defn{$1/4$-\trickyalignedevenfrag{$({=}k)$}} if it is $1/4$-tricky and breaks each request into $k$ equal-sized fragments. Recall that an algorithm is said to be \defn{aligned} if it allocates each request $r$ (whose size is, WLOG, a power of two) to a region of the form $((j - 1)|r|, j|r|]$ for some non-negative integer $j$. We now prove that we can reduce to algorithms which are aligned. 
\begin{lemma}\label{lem:aligned}
    Let $\alg$ be a $1/4$-\tricky algorithm that breaks each request into exactly $k = 2^i$ equal-sized fragments. Let $\mathcal{W}$ be the set of workloads in which every request is a power-of-$2$ size. 
    Then there exists a $1/4$-\trickyalignedevenfrag{$({=}2^i)$} algorithm, $\balg$, such that for all  $\workload \in \mathcal{W}$, $\hwm(\balg,\workload) \leq \hwm(\alg, 2\workload)$, where $2\workload$ refers to the workload which allocates requests of exactly $2$ times the size of those in $\workload$.
\end{lemma}

\begin{proof}
    We will simulate $\alg$ and use its output to create $\balg$.
    
    To allocate a request $q$ of size $\size{q}$ in $\balg$, we simulate a request $q'$ of size $2\size{q}$ in $\alg$.
    If $\alg$ ignores the request, then $\balg$ does too.
    Otherwise, for an allocated fragment $F$ of $q'$,
    the center of $F$ must fall within some $\size{q}/2$-aligned block $B$ (possibly on its edge).
    Then $F$ must contain $B$, since it is twice as large as $B$.
    $\balg$ then allocates into $B$.

    $\balg$ is $1/4$-\tricky and \alignedevenfrag{$({=}2^i)$}, because $\alg$ is and by construction.
    $\balg$'s allocations are a subset of $\alg$'s, so $\hwm(\balg,\workload) \leq \hwm(\alg, 2\workload)$.
\end{proof}

Combining \Cref{lem:tricky,lem:equal-size,lem:aligned} show that \Cref{thm:main-lower} with the WLOG assumptions that $k \ge 2$ is a power of $2$ and $M, R$ are powers of $k$, we can conclude that \Cref{thm:main-lower} reduces to \Cref{thm:main-lower-reduced}.

\end{appendix}

\end{document}